\documentclass[preprint,12pt]{elsarticle}

\usepackage{amsmath,amssymb,amsthm,amsfonts,latexsym,graphicx,subfigure}
\usepackage{mathrsfs}
\usepackage[dvips]{color}
\usepackage{indentfirst}
\usepackage{fancyhdr}
\usepackage{bbm}
\usepackage{accents}

\biboptions{sort&compress}

\makeatletter
\def\wbar{\accentset{{\cc@style\underline{\mskip8mu}}}}
\makeatother

\makeatletter
\def\wbar{\accentset{{\cc@style\underline{\mskip8mu}}}}

\makeatother

\renewcommand{\vec}[1]{\mbox{\boldmath \small $#1$}}

\def\mi{\mathtt{i}}
\def\numzero{\#^z}
\def\numextr{\#^e}
\def\sech{\mathrm{sech}} 
\def\me{\mathrm{e}} 
\newcommand{\pp}[2]{\frac{\partial{#1}}{\partial{#2}}}
\newcommand{\dd}[2]{\frac{\dif{#1}}{\dif{#2}}}
\newcommand{\bmb}[1]{\left(#1\right)}
\newcommand{\bbmb}[1]{\left[#1\right]}

\def\dif{\mathrm{d}}

\newcommand{\diag}{\mathrm{diag}}
\newcommand{\imag}{\mathrm{Im}}
\newcommand{\real}{\mathrm{Re}}

\newtheorem{Theorem}{Theorem}

\newtheorem{Lemma}[Theorem]{Lemma}
\newtheorem{Remark}{Remark}
\theoremstyle{definition}

\newcommand{\ansatz}{\textit{ansatz}{ }}
\newcommand{\videpost}{\textit{vide post}{}}
\newcommand{\etc}{\textit{etc}{}}

\newcommand{\ie}{\textit{i.e.}{~}}
\newcommand{\eg}{\textit{e.g.}{~}}

\newcommand{\vgamma}{\vec{\gamma}}

\newcommand{\vLambda}{\vec{\Lambda}}
\newcommand{\vPsi}{\vec{\Psi}}

\newcommand{\vpsi}{\vec{\psi}}

\newcommand{\vS}{\vec{S}}

\begin{document}

\begin{frontmatter}

\title{Multi-hump solitary waves of nonlinear Dirac equation}
\author[ss]{Jian Xu}

\author[ss]{Sihong Shao\corref{cor}}
\ead{sihong@math.pku.edu.cn}

\author[thz]{Huazhong Tang}

\author[wdy]{Dongyi Wei}

\cortext[cor]{Corresponding author.}

\address[ss]{LMAM and School of Mathematical Sciences, Peking University,
Beijing 100871, China}
\address[thz]{HEDPS, CAPT \& LMAM, School of Mathematical Sciences, Peking University,
Beijing 100871, China}
\address[wdy]{School of Mathematical Sciences, Peking University,
Beijing 100871, China}
\begin{abstract}
This paper concentrates on a (1+1)-dimensional nonlinear Dirac (NLD)
equation with a general self-interaction, being a linear combination
of the scalar, pseudoscalar, vector and axial vector
self-interactions to the power of the integer $k+1$. The solitary
wave solutions to the NLD equation are analytically derived, and the
upper bounds of the hump number in the charge, energy and momentum
densities for the solitary waves are proved in theory. The results
show that: (1) for a given integer $k$, the hump number in the
charge density is not bigger than $4$, while that in the energy
density is not bigger than $3$; (2) those upper bounds can only be
achieved in the situation of higher nonlinearity, namely,
$k\in\{5,6,7,\cdots \}$ for the charge density and
$k\in\{3,5,7,\cdots\}$ for the energy density; (3) the momentum
density has the same multi-hump structure as the energy density; (4)
more than two humps (resp. one hump) in the charge (resp. energy)
density can only happen under the linear combination of the
pseudoscalar self-interaction and at least one of the scalar and
vector (or axial vector) self-interactions. Our results on the
multi-hump structure will be interesting in the interaction dynamics
for the NLD solitary waves.
\end{abstract}
\end{frontmatter}

\section{Introduction}
\label{sec:intro}

There is a remarkable upsurge of interest in the nonlinear Dirac
(NLD) models or equations, as they emerge naturally as practical
models in many physical systems, such as the extended particles in
particle physics
\cite{Ivanenko1938,FinkelsteinLelevierRuderman1951,FinkelsteinFronsdalKaus1956,Heisenberg1957},
the gap solitons in nonlinear optics \cite{Barashenkov1998},
Bose-Einstein condensates in honeycomb optical lattices
\cite{Haddad2009}, phenomenological models of quantum chromodynamics
\cite{Fillion-Gourdeau2013}, as well as matter influencing the
evolution of the Universe in cosmology \cite{Saha2012}. To make the
resulting NLD model to be Lorentz invariable, the so-called
self-interaction Lagrangian can be built up from the bilinear
covariants which are categorised into five types: scalar,
pseudoscalar, vector, axial vector and tensor. Different
self-interactions give rise to different NLD models. Several
interesting models have been proposed and investigated based on the
scalar bilinear covariant in
\cite{Gursey1956,Soler1970,GrossNeveu1974,Mathieu1984}, on the
vector bilinear covariant in \cite{Thirring1958}, on the axial
vector bilinear covariant \cite{Mathieu1983}, on both scalar and
pseudoscalar bilinear covariants \cite{RanadaRanada1984}, on both
the scalar and vector bilinear covariants
\cite{Stubbe1986-jmp,NogamiToyama1992} \etc. All of these models
have attracted wide interest of physicists and mathematicians,
especially on looking for solitary wave solutions and investigating
their physical and mathematical properties.

A key feature of the NLD equations is that
it allows solitary wave solutions or particle-like solutions --- the stable localized solutions
with finite energy and charge \cite{Ranada1983}.
That is, the particles appear as intense localized regions of field
which can be recognized as the basic ingredient
in the description of extended objects in quantum field theory \cite{Weyl1950}.
For the NLD equation in (1+1) dimensions (\ie one time dimension plus one space dimension),
several analytical solitary wave solutions are derived
in \cite{LeeKuoGavrielides1975} 
for the quadric nonlinearity,
\cite{Mathieu1985-jpa-mg,Mathieu1985-prd} for fractional
nonlinearity as well as
\cite{Stubbe1986-jmp,CooperKhareMihailaSaxena2010} for general
nonlinearity by using explicitly the constraint resulting from the
energy-momentum conservation. Most existing studies on the NLD
solitary waves focus on the situation with the self-interaction
Lagrangian constructed from one single bilinear covariant. For
example, the Soler model \cite{Soler1970} and the Thirring model
\cite{Thirring1958} involve respectively the quadric scalar
self-interaction and the quadric vector self-interaction, and
further discussion about extending such two models to the situation
with the integer nonlinearity is recently presented in
\cite{CooperKhareMihailaSaxena2010,MertensQuinteroCooper2012}.

With the help of the analytical expressions of the NLD solitary wave solutions,
the interaction dynamics among them can further be conveniently studied
and rich nonlinear phenomena have been revealed in a series of work \cite{AlvarezCarreras1981,ShaoTang2005,ShaoTang2006,ShaoTang2008,XuShaoTang2013}.
An important step in this direction has been made by
Alvarez and Carreras \cite{AlvarezCarreras1981},
who simulated numerically the interaction dynamics
for the (1+1)-dimensional NLD solitary waves under the quadric scalar self-interaction (\videpost)
\cite{AlvarezKuoVazquez1983}.
Shao and Tang have revisited this interaction dynamics problem \cite{ShaoTang2005}
by employing a higher-order accurate method.
They not only reproduced the phenomena observed by Alvarez and Carreras but also observed
that collapse happens in binary and ternary collisions of two-hump NLD solitary waves.
Very recently, they have further investigated the interaction dynamics for the NLD solitary
waves under the linear combination of scalar and vector self-interactions with the integer nonlinearity
and revealed that the interaction dynamics depend on the power exponent
 of the self-interaction in the NLD equation,
for example,
collapse happens again after collision of two equal one-hump NLD solitary waves
under the cubic vector self-interaction in contrast to no collapse scattering
for corresponding quadric case \cite{XuShaoTang2013}.
Their numerical results inferred that both the multi-hump (two-hump) profile and
high order nonlinearity could undermine
the stability during the scattering of the NLD solitary waves.
Note in passing that,
though multi-hump solitary waves have been found for many other nonlinear models,
see \eg \cite{OstrovskayaKivsharSkryabinFirth1999} and references therein,
the detailed study of the multi-hump solitary waves of the NLD model lacks except
for the two-hump structure first pointed out by Shao and Tang \cite{ShaoTang2005}
and later gotten noticed by other researchers, see \eg \cite{CooperKhareMihailaSaxena2010}.

In (1+1) dimensions,
the pseudoscalar and tensor bilinear covariants are linearly dependent,
and a direct generalization of self-interaction
is linearly combining the scalar,
pseudoscalar, vector and axial vector bilinear covariants with arbitrary nonlinearity,
called by the linear combined self-interaction (see Eq.~\eqref{generalLI}).
A natural question is raised here:
How about the interaction dynamics
for the (1+1)-dimensional NLD solitary waves under such linear combined self-interaction?
In answering the question,
on the one hand,
efficient and stable numerical methods are necessary in order to solve accurately the NLD equation with
the linear combined self-interaction in long time simulations.
Actually, we have demonstrated that
both the Runge-Kutta discontinuous Galerkin method and
the exponential operator splitting method are fit for the job \cite{ShaoTang2006,XuShaoTang2013}.
On the other hand,
more detailed information on the
physical and mathematical properties
of the NLD solitary waves under the linear combined self-interaction
is essential to investigating the interaction dynamics.
The present work will focus on studying these properties and
try to answer questions such as:
How to choose the coefficients of the linear combined self-interaction
to make the NLD model be physically significant and have solitary wave solutions?
What parameters does the multi-hump structure depend on?
Is the hump number related to the power exponent of the self-interaction?

The paper is organized as follows. The NLD
equation with the linear combined self-interaction
is introduced in Section \ref{sec:nld} and
the range of the combination coefficients
is also determined there by the Hermiticity requirement of the self-interaction.
In Section \ref{sec:solitary},
the localized solitary wave solutions are analytically derived with the help of
the conservation laws.
The multi-hump structure of the charge, energy and momentum densities is analyzed in Section \ref{sec:multi-hump}
where the upper bounds of the hump number in those densities are proved in theory.
The paper is concluded in Section \ref{sec:conclusion} with a few remarks.

\section{Nonlinear Dirac equation}
\label{sec:nld}

This section will introduce the (1+1)-dimensional NLD equation
with the linear combined self-interaction. Throughout the paper,
units in which both the speed of light and the reduced Planck constant are equal
to one will be used,
and the Einstein summation convection will be applied,
namely, summation over repeated indices.
The NLD equation has the following general covariant form
\begin{equation}
(\mi\vgamma^\mu\partial_\mu-m)\vPsi+\pp{L_{\text{I}}}{\wbar{\vPsi}}=0, \label{generalnld}
\end{equation}
being the Euler-Lagrange equation
$\partial_\mu \big(\partial{L}/\partial (\partial_\mu\wbar{\vPsi})\big) -
\partial{{L}}/\partial{\wbar{\vPsi}}=0$,
where
the spinor $\vPsi$ has two complex components,
$\wbar{\vPsi}=\vPsi^\dag\vgamma^0$ with superscript $\dag$ denoting the conjugate transpose,
$\vgamma^\mu$ are Gamma matrices (we choose $\vgamma^0=\sigma^z$ and $\vgamma^1=\mi\sigma^y$ as did in \cite{AlvarezCarreras1981,ShaoTang2005} where $\sigma^{x,y,z}$ are the standard Pauli matrices) and $\partial_\mu=\pp{}{x^\mu}$ for $\mu =0,1$,
$\mi$ is the imaginary unit,
$m\geq 0$ is the mass (the NLD model is called massive if $m > 0$ and massless if $m=0$),
and the Lagrangian $L$ reads
\begin{equation}\label{total_lag}
{L}=L_{\text{D}}+{L}_{\text{I}}.
\end{equation}
Here $L_\text{D}$ denotes the Dirac Lagrangian given by 
\begin{equation} \label{L_D}
L_{\text{D}}=\frac{\mi}2\big(\wbar{\vPsi}\vgamma^\mu\partial_\mu\vPsi - (\partial_\mu\wbar{\vPsi})
\vgamma^\mu\vPsi\big)-m\wbar{\vPsi}\vPsi,
\end{equation}
and the self-interaction Lagrangian ${L}_{\text{I}}$ is
a nonlinear functional of the spinors $\vPsi$ and $\wbar{\vPsi}$ but independent of $\partial_\mu\wbar{\vPsi}$,
\eg a general linear combined self-interaction in \eqref{generalLI} will be considered in this work.
Physically,
the self-interaction Lagrangian ${L}_{\text{I}}$ is not only required to be invariant under the Lorentz transformation
(see Eq.~\eqref{eq:boost2}) but also should be carefully chosen such that the resulting solution $\vPsi$ to the NLD equation \eqref{generalnld} satisfies
the following conservation laws
\begin{align}
\partial_\mu j^\mu &=0, \label{jconserve}\\
\partial_\mu{T}^{\mu\nu} &=0, \label{Tconserve}
\end{align} 
where the current vector $j^\mu$ and the energy-momentum tensor $T^{\mu\nu}$ are defined respectively as
\begin{align}
j^\mu &= \wbar{\vPsi}\vgamma^\mu\vPsi, \label{jvector}\\
T^{\mu\nu} &=\frac{\mi}2\big(\wbar{\vPsi}\vgamma^\mu\partial^\nu\vPsi
-(\partial^\nu\wbar{\vPsi})\vgamma^\mu\vPsi\big)-\eta^{\mu\nu}{L}. \label{Ttensor}
\end{align}
Here $\partial^\mu=\eta^{\mu\nu}\partial_\nu$ and $\eta^{\mu\nu}=\eta_{\mu\nu}=\diag(1,-1)$ is the Minkowski metric.
Eq.~\eqref{jconserve} corresponds to the mass conservation,
while Eq.~\eqref{Tconserve} gives
the energy conservation for $\nu=0$
and the momentum conservation for $\nu=1$.
According to Eqs.~\eqref{jconserve}, \eqref{Tconserve},
\eqref{jvector} and \eqref{Ttensor},
for the localized spinor $\vPsi$,
integrating the zero components of the current vector $j^\mu$
and the energy-momentum tensor ${T}^{\mu\nu}$ yields
three conserved quantities, \ie the charge $Q$, the energy $E$, and the momentum $P$,
as follows
\begin{align}
Q &= \int_{-\infty}^{+\infty} j^0 \dif x,
\label{Q-2C}\\
E &= \int_{-\infty}^{+\infty} {T}^{00} \dif x,
\label{E-2C}\\
P &= \int_{-\infty}^{+\infty} {T}^{01} \dif x.
\label{P-2C}
\end{align}

As we have mentioned in Section~\ref{sec:intro},
the self-interaction Lagrangian ${L}_{\text{I}}$
can be built up from the bilinear covariants
and several NLD models exist in the literature corresponding to different bilinear covariants.
There are five types of bilinear covariants:
the scalar bilinear covariant is $\wbar{\vPsi}\vPsi$,
the pseudoscalar bilinear covariant is $\wbar{\vPsi}\vgamma^5\vPsi$,
the vector bilinear covariant is $\wbar{\vPsi}\vgamma^\mu\vPsi$,
the axial vector bilinear covariant is $\wbar{\vPsi}\vgamma^\mu\vgamma^5\vPsi$,
and
the tensor bilinear covariant is
$\frac{\mi}2 \wbar{\vPsi}(\vgamma^\mu\vgamma^\nu-\vgamma^\nu\vgamma^\mu)\vPsi$,
where $\vgamma^5 = \vgamma^0\vgamma^1$.
In (1+1) dimensions,
it can easily be shown that the tensor and pseudoscalar bilinear covariants are linearly dependent,
\eg $\frac{\mi}2 \wbar{\vPsi}(\vgamma^1\vgamma^0-\vgamma^0\vgamma^1)\vPsi=-\mi \wbar{\vPsi}\vgamma^5\vPsi$,
and then the remaining four bilinear covariants are used to
construct the following self-interactions
\begin{align}
L_{\text{S}} &= \wbar{\vPsi}\vPsi= |\Psi_1|^2 - |\Psi_2|^2 \in\mathbb{R}, \label{S_2c}\\
L_{\text{P}} &= -\mi\wbar{\vPsi}\vgamma^5\vPsi=2 \imag (\Psi_1^\ast\Psi_2) \in\mathbb{R}, \label{iP_2c}\\
L_{\text{V}} &= \wbar{\vPsi}\vgamma^\mu\vPsi \wbar{\vPsi}\vgamma_\mu\vPsi, \label{V_mu}\\
L_{\text{A}} &= \wbar{\vPsi}\vgamma^\mu\vgamma^5\vPsi \wbar{\vPsi}\vgamma_\mu\vgamma^5\vPsi, \label{A_mu}
\end{align}
where $\Psi_1$ and $\Psi_2$ are two components of the spinor $\vPsi$,
the superscript $\ast$ denotes the complex conjugate,
and $\vgamma_\mu=\eta_{\mu\nu}\vgamma^\nu$ are the covariant Gamma matrices.
Further more, direct calculation shows
the relation between $L_{\text{V}}$ and $L_{\text{A}}$
\begin{equation}\label{L_V}
    L_{\text{V}} = - L_{\text{A}} =\bmb{|\Psi_1|^2 + |\Psi_2|^2}^2 - \bmb{2 \real(\Psi_1^\ast\Psi_2)}^2 \geq 0,
\end{equation}
and thus the general linear combined self-interaction can be formally written as 
\begin{equation}\label{generalLI}
L_\text{I} = s(L_{\text{S}}) ^{k+1} + p(L_{\text{P}}) ^{k+1} + v (L_{\text{V}})^{\frac12(k+1)},
\end{equation}
where the exponent power $k+1$ is an integer, and the linear combination coefficients $s,p,v$
should be carefully chosen such that the resulting NLD models are physically meaning (\videpost).
For some special choice of the parameters $k,s,p,v$, Eq.~\eqref{generalLI} will reduce
to the often-cited NLD models in literature
such as the Thirring model \cite{Thirring1958,Coleman1975}
and the Soler or Gross-Neveu model \cite{Soler1970,GrossNeveu1974}.
If the spinor $\vPsi$ is scaled by a constant factor as $\widetilde{\vPsi}=\sqrt{\alpha}\vPsi$
with $\alpha\in\mathbb{C}$,
then the scaled self-interaction Lagrangian will be
$\alpha^{k+1}L_\text{I}[\vPsi]$ which shows that the power exponent to $\alpha$ is $k+1$.
In such sense, we call that the self-interaction Lagrangian $L_\text{I}$ has the power exponent $k+1$,
for example, the quadric and cubic self-interaction  are referred to the case $k=1$
and the case $k=2$, respectively.
The linear combination of the quadric scalar and quadric pseudoscalar self-interactions
has been studied in \cite{LeeKuoGavrielides1975,RanadaRanada1984},
while the linear combination of the scalar and vector self-interactions
with a general power exponent has been considered in \cite{Stubbe1986-jmp,XuShaoTang2013}.
The linear combined self-interaction \eqref{generalLI} with $k=1$
has also been mentioned in \cite{Ranada1977}.

Obviously,
the linear combined self-interaction ${L}_{\text{I}}$ in \eqref{generalLI} is Lorentz invariant
because each of $L_{\text{S}}, L_{\text{P}}, L_{\text{V}}$ is
invariant under the Lorentz transformation.
Accordingly,
the only remaining physical requirement is
to choose the linear combination coefficients in \eqref{generalLI}
such that the resulting NLD spnior $\vPsi$
satisfies the conservation laws \eqref{jconserve} and \eqref{Tconserve}.
It can readily be shown that
the linear combined self-interaction \eqref{generalLI}
satisfies the so-called homogeneity relation \cite{Mathieu1985-jpa-mg,Stubbe1986-jmp}
\begin{equation}
\wbar{\vPsi}\pp{{L}_{\text{I}}}{\wbar{\vPsi}}=(k+1){L}_{\text{I}}, \quad
\bmb{\pp{{L}_{\text{I}}}{\wbar{\vPsi}}}^\dag\vgamma^0\vPsi=(k+1){L}_{\text{I}}^\ast. \label{ppL=L}
\end{equation}
Combining Eqs.~\eqref{generalnld}, \eqref{jvector} and \eqref{ppL=L}
gives that the conservation law \eqref{jconserve} is equivalent to
the Hermiticity of the linear combined self-interaction \eqref{generalLI},
\ie
\begin{equation}
  \label{eq:LI-real}
  L_{\text{I}} = L_{\text{I}}^\ast,
\end{equation}
which poses a requirement the self-interaction \eqref{generalLI} must satisfy.
Multiplying Eq.~\eqref{generalnld} from left with $\wbar{\vPsi}$
plus the conjugate transpose of Eq.~\eqref{generalnld} multiplied with $\vgamma^0{\vPsi}$ from right
yields
\begin{equation*}
2 L_\text{D} + \wbar{\vPsi}\pp{{L}_{\text{I}}}{\wbar{\vPsi}}+
\bmb{\pp{{L}_{\text{I}}}{\wbar{\vPsi}}}^\dag\vgamma^0\vPsi=0,
\end{equation*}
for Eq.~\eqref{L_D},
and then we have a useful relation
\begin{equation} \label{LLIrelation}
{L} = -k {L}_{\text{I}}.
\end{equation}
for Eqs.~\eqref{total_lag}, \eqref{ppL=L} and \eqref{eq:LI-real}.
In consequence, combining the homogeneity relation \eqref{ppL=L},
the relation \eqref{LLIrelation} between $L$ and $L_{\text{I}}$
as well as the Hermiticity requirement \eqref{eq:LI-real},
and direct algebraic calculation leads to the conservation law \eqref{Tconserve}.
That is,
the Hermiticity requirement \eqref{eq:LI-real}
is the only condition for making
the NLD spinor $\vPsi$ under the linear combined self-interaction \eqref{generalLI}
follows the conservation laws \eqref{jconserve} and \eqref{Tconserve}.
Below we will use the Hermiticity requirement \eqref{eq:LI-real}
to choose the linear combination coefficients in \eqref{generalLI}.
Before that,
we would like to make a remark that
the cases of $k=-1$ and $k=0$ will not be considered in the following
because the NLD equation \eqref{generalnld} degenerates to the linear Dirac equation
when $k=-1$ according to Eq.~\eqref{ppL=L},
and the Lagrangian $L$ vanishes (\ie $L\equiv 0$)
when $k=0$ for the relation \eqref{LLIrelation}.

The Hermiticity condition \eqref{eq:LI-real} implies
\begin{equation}
  \label{eq:constraint-spva}
  (s-s^\ast) (L_{\text{S}})^{k+1}+ (p-p^\ast) (L_{\text{P}})^{k+1}+
(v-v^\ast)(L_{\text{V}})^{\frac12(k+1)}
 =0,
\end{equation}
for $k\in\mathbb{Z}\setminus \{-1,0\}$ and
$L_{\text{S}}, L_{\text{P}}, L_{\text{V}}$
are all real as shown in Eqs.~\eqref{S_2c}, \eqref{iP_2c} and \eqref{L_V}.
In particular, for the quadric case (\ie $k=1$),
Eq.~\eqref{eq:constraint-spva}
further reduces to
\begin{equation}\label{eq:constraint-spva-k=1}
\bmb{s-p-s^\ast+p^\ast}(L_{\text{S}})^2 + \bmb{v+p-v^\ast-p^\ast} L_{\text{V}}=0,
\end{equation}
on account of $(L_{\text{P}})^2 = L_{\text{V}} - (L_{\text{S}})^2$ \cite{Ranada1977}.
Because of the arbitrariness of the NLD spinor $\vPsi$,
Eq.~\eqref{eq:constraint-spva-k=1} implies that both $s-p$ and $v+p$ must be real when $k=1$,
otherwise $s, p, v$ must all be real for $k\in\mathbb{Z} \setminus\{0, \pm 1\}$.
The range of the parameters $\{s, p, v\}$
in the linear combined self-interaction \eqref{generalLI}
with a given integer power exponent $k+1$ reads as follows
\begin{equation}\label{ek}
\mathcal{E}_k :=
\begin{cases}
  \{(s,p,v)| s-p\in\mathbb{R},
  v+p\in\mathbb{R},
  |s-p| + |v+p| \neq 0\} &\text{for}\quad k=1,\\
   \{(s,p,v)| s\in\mathbb{R}, p\in\mathbb{R},
  v\in\mathbb{R}, |s|+|p|+|v|\neq 0\}&\text{for}\quad
 k\in\mathbb{Z}\setminus \{0,\pm 1\},
\end{cases}
\end{equation}
where the coefficients with which $L_\text{I}\equiv 0$ holds have been excluded.

In the next section,
for $k\in\mathbb{Z}\setminus \{-1,0\}$,
we are going to look for the localized solitary wave solutions
for the NLD equation \eqref{generalnld} with the linear combined self-interaction \eqref{generalLI}
of a given integer power exponent $k+1$
under that the linear combination coefficients in Eq.~\eqref{generalLI} belong to $\mathcal{E}_k$ in Eq.~\eqref{ek}.

\section{Solitary wave solutions}
\label{sec:solitary}

This section will focus on
seeking the localized solutions of the following form
for the (1+1)-dimensional NLD equation \eqref{generalnld} with \eqref{generalLI}
in the spirit of the methods used in \cite{LeeKuoGavrielides1975,ChangEllisLee1975,Mathieu1985-jpa-mg,Stubbe1986-jmp}.
The solution with the form
\begin{equation}\label{localsol1}
\vPsi(x, t)
=\me^{-\mi \omega t} \vpsi(x),
\quad
\vpsi(x) =
\left(
\begin{matrix}
\varphi(x)\\
\chi(x)
\end{matrix}
\right)
\end{equation}
is wanted,
where $\omega \geq 0$ is the frequency,
and
both $|\varphi(x)|$ and $|\chi(x)|$ are required to
decay very fast to zero as $|x|\rightarrow +\infty$.
Such solution is said to be localized in $\mathbb{R}$ as mentioned before.
Substituting the \ansatz \eqref{localsol1} into the Lagrangian \eqref{total_lag}
and the energy-momentum tensor \eqref{Ttensor} gives respectively
\begin{align}
{L} &= \omega \wbar{\vpsi}\vgamma^0\vpsi
+\frac{\mi}2(\wbar{\vpsi}\vgamma^1\partial_x\vpsi - (\partial_x\wbar{\vpsi})\vgamma^1\vpsi)-m\wbar{\vpsi}\vpsi+{L}_{\text{I}},
\label{total_lag_ti}\\
{T}^{00} &= - \frac{\mi}2 (\wbar{\vpsi}\vgamma^1\partial_x\vpsi -(\partial_x\wbar{\vpsi})\vgamma^1\vpsi)
+m\wbar{\vpsi}\vpsi-{L}_{\text{I}}, \label{t_00-sw} \\
{T}^{01} &= - \frac{\mi}2(\wbar{\vpsi}\vgamma^0\partial_x\vpsi -(\partial_x\wbar{\vpsi})\vgamma^0\vpsi),  \label{t_01-sw}\\
{T}^{10} &= \omega \wbar{\vpsi}\vgamma^1\vpsi, \label{t_10-sw}\\
{T}^{11} &= \omega \wbar{\vpsi}\vgamma^0\vpsi-m\wbar{\vpsi}\vpsi+{L}_{\text{I}}, \label{t_11-sw}
\end{align}
all of which are independent of the time $t$ in this moment,
and thus the conservation law \eqref{Tconserve} becomes
\begin{equation*}
\frac{\dif{T}^{10}}{\dif x}=\frac{\dif{T}^{11}}{\dif x}=0,
\end{equation*}
which implies that
\begin{equation}\label{t00=t11=0}
{T}^{10}={T}^{11}=0
\end{equation}
for the localized solutions \eqref{localsol1}, \ie
\begin{align}
\omega \wbar{\vpsi}\vgamma^1\vpsi&=0, \label{t10} \\
\label{key-Li}
\omega \vpsi^\dag\vpsi- m\wbar{\vpsi}\vpsi + {L}_{\text{I}}&=0.
\end{align}
To ensure Eq.\ \eqref{t10}, we require
\begin{equation}\label{imag}
\wbar{\vpsi}\vgamma^1\vpsi
=\varphi^\ast\chi+\varphi\chi^\ast=0.
\end{equation}
That is, $\varphi^\ast\chi$ is imaginary. To this end,
we assume that $\varphi$ is real and $\chi$ is imaginary as follows
\begin{equation}
\left(
\begin{matrix}
\varphi(x)\\
\chi(x)
\end{matrix}
\right)=
R(x) \left( \begin{matrix}
  \cos\bmb{\theta(x)} \\  \mi \sin\bmb{\theta(x)}\end{matrix}\right),\label{phaseexp}
\end{equation}
where both $R(x)$ and $\theta(x)$ are under-determined real functions.
Only the classical solutions are considered below
and both $R(x)$ and $\theta(x)$ at least belong to $C^1(\mathbb{R})$
which consists of all differentiable functions in $\mathbb{R}$ whose derivative is continuous.
Meanwhile,
we assume that for any $x\in\mathbb{R}$,
the charge density $j_0(x)$ does not vanish,
that means,
the particle described by the NLD spinor $\vPsi$ in Eq.~\eqref{localsol1}
has a positive probability to go anywhere in $\mathbb{R}$.
Under such assumption,
according to Eq.~\eqref{jvector},
for any $x\in\mathbb{R}$,
we have $R(x)\neq 0$ for
\begin{equation}
\rho_Q(x) := j^0[\vPsi](x,t) = \vpsi^\dag(x)\vpsi(x) =  (R(x))^2>0, \label{j0R2}
\end{equation}
where the spinor $\vPsi$ is given in Eq.~\eqref{localsol1}
and the notation $\rho_Q(x)$ denoting the charge density has been
introduced for convenience.
Moreover,
physically, the charge $Q$ defined in Eq.~\eqref{Q-2C} is required to be finite,
\ie $Q=\int_{-\infty}^{+\infty} \rho_Q(x) \dif x < +\infty$.
Note in passing that substituting Eq.~\eqref{phaseexp} into Eq.~\eqref{t_01-sw}
directly leads to
\begin{equation}\label{T01=0}
T^{01} = 0
\end{equation}
which means the momentum density vanishes for all $x\in\mathbb{R}$.
Further substituting Eqs.~\eqref{localsol1} and \eqref{phaseexp} into Eqs.~\eqref{S_2c},
\eqref{iP_2c} and \eqref{V_mu} leads to, respectively,
\begin{align}
L_{\text{S}}      &= \wbar{\vpsi}(x) \vpsi(x) = (R(x))^2\cos(2\theta(x)), \label{S_theta}\\
L_{\text{P}}      &=  (R(x))^2\sin(2\theta(x)), \label{P_theta} \\
L_{\text{V}}      &= (R(x))^4, \label{VV_theta}
\end{align}
and then the linear combined self-interaction \eqref{generalLI} becomes
\begin{equation}\label{generalLI-3}
L_\text{I} = (R(x))^{2(k+1)} G(x),
\end{equation}
where
\begin{equation}\label{G(theta)}
G(x):=s (\cos\bmb{2\theta(x)})^{k+1} + {p} (\sin\bmb{2\theta(x)})^{k+1} + {v},
\end{equation}
are introduced for convenience.

Combining Eqs.~\eqref{LLIrelation}, \eqref{total_lag_ti} and \eqref{key-Li} yields
\begin{equation}
k \omega \vpsi^\dag\vpsi-k m\wbar{\vpsi}\vpsi
-\frac{\mi}2 (\wbar{\vpsi}\vgamma^1 \vpsi_x -\wbar{\vpsi}_x\vgamma^1 \vpsi)=0,
\label{key-eq1}
\end{equation}
which does not depend on the particular type of the self-interaction involved and could be solved analytically.
Substituting the \ansatz \eqref{phaseexp} into \eqref{key-eq1} gives rise to
the ordinary differential equation
\begin{equation}
\frac{\dif \theta(x)}{m\cos\bmb{ 2\theta(x)}-\omega} = k \dif x, \label{theta-recast1}
\end{equation}
under the condition of $m\cos\bmb{ 2\theta(x)}-\omega\neq 0$,
otherwise $\theta(x)=\frac12\cos^{-1}\frac{\omega}{m}$. According to the integral formula
\begin{equation}
\int_{u_0}^{u} \frac{\dif \theta}{a+b \cos (c\theta)} \quad (c\neq 0)=
\begin{cases}
  {\frac{2}{c \sqrt{b^2-a^2}} \tanh^{-1}\left(\frac{(b-a)}{\sqrt{b^2-a^2}}
      \tan \left(\frac{c u}{2}\right)\right)}, & |b|>|a|, \\
  -\frac{1}{a c}\cot \left(\frac{c u}{2}\right), & b=-a\neq 0,\\
  {\frac{2}{c \sqrt{a^2-b^2}} \tan^{-1}\left(\frac{(a-b)}{\sqrt{a^2-b^2}}
      \tan \left(\frac{c u}{2}\right)\right)}, & |b|<|a|, \\
  \frac{1}{a c} \tan \left(\frac{c u}{2}\right), & b=a\neq 0,
\end{cases}\label{integral}
\end{equation}
where the value of $u_0$ is taken to be $\pi/c$ for the second case and to be zero for the remaining three cases,
the solution of \eqref{theta-recast1} can be obtained as follows:
\begin{itemize}
\item When $m > \omega \geq 0$, the solution of \eqref{theta-recast1} with $\theta(0) = 0$ is
\begin{equation} \label{theta-1}
\theta(x) = \tan^{-1} ( \alpha \tanh(k\beta x)) \in \left(-\tan^{-1}(\alpha), \tan^{-1}(\alpha)\right)
\subseteq \left(-\frac{\pi}{4},\frac{\pi}{4}\right),
\end{equation}
where
\begin{equation}\label{theta-1b}
\alpha = \sqrt{\frac{m-\omega}{m+\omega}},\quad \beta = \sqrt{m^2-\omega^2}.
\end{equation}

\item When $\omega = m > 0$, the solution of \eqref{theta-recast1} with $\theta(0) = \frac{\pi}2$ is
\begin{equation}   \label{eq:w=m}
  \theta(x) = \cot^{-1}(2mkx) \in \left(0, \pi\right).
\end{equation}

\item When $\omega > m \geq 0$, the solution of \eqref{theta-recast1} with $\theta(0) = 0$ is
\begin{equation}
\label{theta-2}
\theta(x) = \tan^{-1} \left( \sqrt{\frac{\omega-m}{\omega+m}} \tan\left(-k\sqrt{\omega^2-m^2} x\right)\right)\in \left(-\frac{\pi}{2},\frac{\pi}{2}\right).
\end{equation}
\end{itemize}
Note in passing that the solution \eqref{theta-2}
can also be reformulated into the solution \eqref{theta-1} using the properties:
$\tanh(\mi x)=\mi \tan(x)$ and $\tan(-x)=-\tan(x)$,
and the last case of the integral formula \eqref{integral} can not happen
in Eq.~\eqref{theta-recast1} for both $m$ and $\omega$ are nonnegative.
The remaining task is to solve $R(x)$.

Further combining Eq.~\eqref{key-Li}, \eqref{j0R2}, \eqref{S_theta} and \eqref{generalLI-3}
yields
\begin{equation}\label{key-Li-theta}
(R(x))^{2k}G(x) = m\cos(2\theta(x))-\omega,
\end{equation}
from which we can conclude that
either $G(x)\equiv 0$ in $\mathbb{R}$ (\ie $L_I\equiv 0$ for Eq.~\eqref{generalLI-3} and will not be considered)
or $G(x)\neq 0$ for all $x\in\mathbb{R}$  holds,
namely, either $\Omega_0$ or $\Omega_1$ equals to $\mathbb{R}$ after denoting $\Omega_0:=\{x|G(x)=0\}$ and $\Omega_1:=\{x|G(x)\neq 0\}$,
and the demonstration is as follows.
For $m\geq 0$ and $\omega\geq 0$,
there are only four cases to investigate:
$m=\omega=0$, $\omega>m\geq0$, $m>\omega\geq 0$ and $m=\omega> 0$.
For all $x\in \mathbb{R}$,
we have $R(x)\neq 0$,
and $m\cos(2\theta(x))-\omega=0$ for $m=\omega=0$ or $m\cos(2\theta(x))-\omega<0$ for $\omega>m\geq 0$,
thus $G(x) = 0$ holds in the case of $m=\omega=0$,
while $G(x)\neq 0$ is true for the case of $\omega>m\geq 0$.
When $m>\omega\geq 0$,
if there exist $x_0 < x_1\in \mathbb{R}$ such that $x_0\in\Omega_0$ and $x_1\in\Omega_1$,
then we have:
on the one hand, from Eq.~\eqref{key-Li-theta}, $\cos(2\theta(x))=\frac{\omega}{m}$ holds for all $x\in[x_0,x_1]\cap \Omega_0$;
on the other hand, from Eq.~\eqref{theta-1},
there exits $M>0$ and $\delta = \frac{1-(\alpha \tanh(k\beta M))^2}{1+(\alpha \tanh(k\beta M))^2} >
\frac{1-\alpha^2}{1+\alpha^2}=\frac{\omega}{m}$ such that
$\cos(2\theta(x))\geq \delta$ holds for all $x\in[x_0,x_1]\cap \Omega_1$.
This contradicts the assumption that $\theta(x)$ as well as $\cos(2\theta(x))$ are continuous in $[x_0, x_1]$.
The discussion on the remaining case of $m=\omega>0$ is similar to that on the case of $m>\omega\geq 0$.
Below we will concentrate on the situation of $G(x) \neq 0$ as well as $m\cos(2\theta(x))-\omega\neq 0$
for all $x\in \mathbb{R}$.
From Eq.~\eqref{key-Li-theta}, $R(x)$ is solved in $\mathbb{R}$  as follows
\begin{equation}
R(x) =  \pm \left(\frac{m\cos\bmb{2\theta(x)} - \omega}{G(x)} \right)^{\frac{1}{2k}}, \label{R(x)}
\end{equation}
which expresses $R(x)$ in terms of $\theta(x)$ for Eq.~\eqref{G(theta)},
while $\theta(x)$ has been solved in Eqs.~\eqref{theta-1}, \eqref{eq:w=m} and \eqref{theta-2}.
Consequently, according to Eq.~\eqref{j0R2}, the charge density becomes
\begin{equation}
\rho_Q(x) =  \left(\frac{m\cos\bmb{2\theta(x)} - \omega}{G(x)} \right)^{\frac{1}{k}}. \label{R2}
\end{equation}
It is worth noting that
 the derivation of the above solution $\vPsi(x, t)$
 given in Eq.~\eqref{localsol1} with Eqs.~\eqref{phaseexp}, \eqref{theta-1}, \eqref{eq:w=m}, \eqref{theta-2} and \eqref{R(x)} is referred to as sufficient and logically complete,
 that is to say,
the above function $\vPsi(x, t)$
satisfies the (1+1)-dimensional NLD equation \eqref{generalnld}
with the self-interaction \eqref{generalLI}.
Its  demonstration is easy through directly substituting  $\vPsi(x, t)$ into
 \eqref{generalnld} with \eqref{generalLI} and some algebraic manipulations.

The physical solutions with which the charge $Q$ is finite (\ie $Q<+\infty$
implying ${\displaystyle \lim_{x\rightarrow\infty}}\rho_Q(x)=0$ must be true)
will be selected from Eqs.~\eqref{theta-1}, \eqref{eq:w=m}, \eqref{theta-2} and \eqref{R(x)}.
On the one hand, we discard the case of $k\in\{-2,-3,-4,\cdots\}$
in which the limit of $\rho_Q(x)$ can not be zero as $x\rightarrow\infty$ and the reason is as follows.
For example,
when $m>\omega\geq0$ and $k<-1$,
the parameter ${p}$ must be zero
(otherwise $G(0)$ will be infinity and then $\rho_Q(x)$ will be discontinuous at $x=0$
according to Eq.~\eqref{R2}), and from Eq.~\eqref{G(theta)}
we have $G(x)=s(\cos(2\theta(x)))^{k+1}+{v}$,
thus ${\displaystyle \lim_{x\rightarrow+\infty}}(m\cos(2\theta(x))-\omega)=0$ for Eq.~\eqref{theta-1}
and
${\displaystyle \lim_{x\rightarrow+\infty}}G(x)=s(\frac{\omega}{m})^{k+1}+{v}$.
In consequence,
from Eq.~\eqref{R2},
it is evident that
$\rho_Q(x)$ will diverge as $x\rightarrow\infty$ if $s(\frac{\omega}{m})^{k+1}+{v} \neq 0$.
In the case of $s(\frac{\omega}{m})^{k+1}+{v} = 0$
which implies that $s$ can not be zero, directly using L'Hospital's rule gives
${\displaystyle \lim_{x\rightarrow+\infty}}\rho_Q(x)=\big(\frac{m}{s(k+1)}\big)^{\frac{1}{k}}\frac{m}{\omega}\neq 0$.
On the other hand, the case of $k\in\mathbb{Z}^+$ and $\omega > m \geq 0$ is also discarded
and the reason is,  in such case,
we have both $m\cos\bmb{2\theta(x)}-\omega \leq m-\omega<0$ and
$|G(x)|\leq |s|+|{p}|+|{v}|$
hold for all $x\in\mathbb{R}$,
thus there exists $\delta = \bmb{\frac{\omega-m}{|s|+|{p}|+|{v}|}}^{\frac1k}>0$
such that $\rho_Q(x) > \delta$ holds for all $x \in \mathbb{R}$.
Therefore, the physical solutions may exist only in the situation of
$k\in\mathbb{Z}^+$ and $m\geq\omega\geq 0$ and will be searched in the following.

\subsection{$k\in\mathbb{Z}^+$ and $m>\omega\geq 0$}
\label{sec:m>w>=0}

This subsection focuses on the situation with $k\in\mathbb{Z}^+$ and $m>\omega\geq 0$,
in which $\theta(x)$ is given in Eq.~\eqref{theta-1} and
monotonously increases from
$-\tan^{-1}(\alpha)$ to $\tan^{-1}(\alpha)$
as $x$ goes from $-\infty$ to $+\infty$.
Then we have
\begin{equation} \label{eq:mcos2-w-m>w}
m\cos 2\theta-\omega =\frac{\alpha\beta \sech^2(k\beta x)}{1+(\alpha \tanh (k\beta x))^2}
 \in (0,m-\omega],
\end{equation}
and thus $1\geq\cos2\theta > \frac{\omega}{m}\geq 0$.
In order to facilitate the subsequent discussion,
we introduce an intermediate function
$y(x)=\tanh(k \beta x)$ which increases monotonously
from $-1$ to $1$ when $x$ goes from $-\infty$ to $+\infty$
and thus $\displaystyle \lim_{x\rightarrow\pm\infty} y(x) = \pm 1$.
The dependence of $y(x)$ and $\theta(x)$ on $x$ is implicitly
implied hereafter.
From Eq.~\eqref{theta-1}, we have the relation
\begin{equation}\label{eq:relation-z-y}
y = \frac{1}{\alpha}\tan(\theta), \quad
\cos\bmb{2\theta} = \frac{1-\alpha ^2y^2}{1+\alpha ^2y^2}, \quad
\sin\bmb{2\theta}=\frac{2\alpha y}{1+\alpha ^2y^2},
\end{equation}
and then rewrite $G(x)$ given in Eq.~\eqref{G(theta)} and $\rho_Q(x)$ given in Eq.~\eqref{R2}
in terms of $y$ into $\tilde{G}(y)$ and $\tilde{\rho}_Q(y)$, respectively,  as follows
\begin{align}\label{tildeG}
\tilde{G}(y) & = s \bmb{\frac{1-\alpha ^2y^2}{1+\alpha ^2y^2}}^{k+1} +
{p} \bmb{\frac{2\alpha y}{1+\alpha ^2y^2}}^{k+1} + {v},\\
\tilde{\rho}_Q(y) & = \bmb{\alpha\beta}^{\frac{1}{k}}
\bmb{\frac{1-y^2}{1+\alpha ^2y^2}}^{\frac1k}\bmb{\frac{1}{\tilde{G}(y)}}^{\frac1k}, \label{tildej0}
\end{align}
and
$\rho_Q(x)>0$ shown in Eq.~\eqref{R2} gives
\begin{equation} \label{G(y)1/k>0}
  \forall y \in(-1,1), \quad \tilde{G}(y) > 0.
\end{equation}
Then the charge $Q$ becomes
\begin{equation}
Q  = \frac{(\alpha\beta)^{\frac{1}{k}}}{k\beta} I(\alpha, k),  \label{eq:Q-sw}
\end{equation}
where
\begin{equation}\label{I-k}
  I(\alpha, k)
  := \int_{-1}^{1}
  \frac{(1-y^2)^{\frac{1}{k}-1} }{\left((1+\alpha ^2y^2)\tilde{G}(y)\right)^{\frac{1}{k}}} \dif y.
\end{equation}
Because $\alpha,\beta\in(0,+\infty)$ and $k\in\mathbb{Z}^+$,
the finite charge condition is equivalent to
\begin{equation}\label{ieq:I<infty}
I(\alpha, k) < \infty,
\end{equation}
and the necessary condition $\displaystyle \lim_{x \rightarrow \infty} \rho_Q(x)=0$ implies
\begin{equation}\label{tildej0=0}
\lim_{y\rightarrow\pm 1} \tilde{\rho}_Q(y) = 0.
\end{equation}

In short, we should seek the solution in the situation with $k\in\mathbb{Z}^+$ and $m>\omega\geq 0$
satisfying the restrictions \eqref{G(y)1/k>0} and \eqref{ieq:I<infty}.
Given $(s,{p},{v})\in\mathcal{E}_k$, the foregoing restrictions are used to determine
the feasible set for $\omega$,
and the discussion is split into two cases:
one is for $k=1$,
the other is for $k\in\{2,3,4,\cdots\}$.

\noindent $\bullet$ When $k=1$,
the inequality \eqref{G(y)1/k>0} becomes:
$\tilde{G}(y)=(s-{p})\bmb{\frac{1-\alpha ^2y^2}{1+\alpha ^2y^2}}^{2}+ {v}+{p}>0$
holds for any arbitrary $y\in(-1,1)$.
It is equivalent to both $\tilde{G}(0)>0$ and $\tilde{G}(1)\geq0$ hold
since $\tilde{G}(y)$ is even with respect to $y\in(-1,1)$
and increases monotonously when $s-{p}\leq 0$ and
decreases monotonously when $s-{p}>0$ as $y$ increases in $[0,1)$.
If $\tilde{G}(1)=0$,
then ${\displaystyle
\lim_{y \rightarrow \pm 1}}\tilde{\rho}_Q(y)
=\frac{\beta(1+\alpha ^2)}{2\alpha\bmb{(s-{p})(1-\alpha ^2)-({v}+{p}) (1+ \alpha ^2)}}\neq 0$,
which violates the necessary condition \eqref{tildej0=0},
and thus we require $\tilde{G}(1) >0$.
Therefore,
for a given $\omega \in\mathcal{F}_{1}$,
there exists $\epsilon = \min\{\tilde{G}(0), \tilde{G}(1)\} >0$
such that
$I(\alpha, 1) \leq \frac{1}{\epsilon}
  \int_{-1}^{1}\frac{1}{1+\alpha^2y^2} \dif y \leq  \frac{2}{\epsilon}<\infty$,
where the set $\mathcal{F}_1$ is define by
\begin{equation}
\mathcal{F}_{1} := \{\omega |  \omega\in[0,m),
\tilde{G}(0)>0, \tilde{G}(1)>0\}. \label{F_k1}
\end{equation}
That is, the feasible set of $\omega$ for the case of $k=1$ is $\mathcal{F}_{1}$.

\noindent $\bullet$
When $k\in\{2,3,4,\cdots\}$,
if $\tilde{G}(1) = 0$, then ${\displaystyle \lim_{y \rightarrow 1}} \tilde{\rho}_Q(y) =
\bmb{\frac{\beta(1+\alpha ^2)^k}{(k+1)(s\alpha(1-\alpha ^2)^k - p\bmb{2\alpha}^{k} -{v}\alpha(1+\alpha ^2)^k)}}^{\frac1k}
\neq 0$, and if $\tilde{G}(-1) = 0$, then
${\displaystyle \lim_{y \rightarrow -1}} \tilde{\rho}_Q(y)$$=
\bmb{\frac{\beta(1+\alpha ^2)^k }{(k+1)(s\alpha(1-\alpha ^2)^k + p\bmb{-2\alpha}^{k} -{v}\alpha(1+\alpha ^2)^k)}}^\frac1k
\neq 0$,
both of which violate the necessary condition \eqref{tildej0=0}.
 Thus, the feasible set for $\omega$ in the case of $k\in\{2,3,4,\cdots\}$
 should be a subset of the following set
 \begin{align*}
   \mathcal{F}_{k} :=& \{\omega |  \omega\in[0,m),
 \forall y \in [-1,1], \tilde{G}(y) > 0 \}, \;\; \text{for} \;\; k\in\{2,3,4,\cdots\}.
 \end{align*}
Since $\tilde{G}(y)$ has at most three extreme points at
$0, \pm \frac{1}{\alpha}\tan\bmb{\frac12 \tan^{-1}\big|\frac{s}{{p}}\big|^{\frac{1}{k-1}}}$ for
$y \in (-1,1)$,
the minimum of $\tilde{G}(y)$ must locate among these three extreme points and the two endpoints.
In consequence, we have equivalently
\begin{equation}\label{F_k2}
  \mathcal{F}_{k}
= \{ \omega |  \omega\in[0,m),
\forall y \in \mathcal{P}, \tilde{G}(y) > 0\}, \;\;\text{for}\;\; k\in\{2,3,4,\cdots\},
\end{equation}
where
$\mathcal{P}
:=\Big\{0, \pm 1, \pm \frac{1}{\alpha}\tan\bmb{\frac12 \tan^{-1}\big|\frac{s}{{p}}\big|^{\frac{1}{k-1}}} \Big\}\cap [-1,1]$ is a finite set of no more than five elements.
It can be readily verified that, for a given $\omega \in\mathcal{F}_{k}$ with $k\in\{2,3,4,\cdots\}$,
there exists $\epsilon = \min_{y \in [-1,1]}\{\tilde{G}(y)\}>0$
such that
$I(\alpha, k)
  \leq \frac{1}{\epsilon^{\frac{1}{k}}}\int_{-1}^{1}
  (1-y^2)^{\frac{1}{k}-1}   \dif y
=  \frac{\sqrt{\pi }\Gamma\bmb{\frac{1}{k}}}{{\epsilon}^{\frac{1}{k}} \Gamma\bmb{\frac{1}{2}+\frac{1}{k}}}
<\infty$,
where $\Gamma(x)$ is the gamma function.
That is, the feasible set of $\omega$ for the case of $k\in\{2,3,4,\cdots\}$
is indeed $\mathcal{F}_{k}$
given in Eq.~\eqref{F_k2}.

\begin{Remark}\rm
Generally,
$\mathcal{F}_k \varsubsetneqq [0, m)$ holds
for most cases of $(s,{p},{v})\in\mathcal{E}_k$ with $k\in\mathbb{Z}^+$.
For the NLD solitary waves with the scalar and vector self-interaction
and $s>0, -s<v\leq0, p=0$,
the feasible set becomes $\mathcal{F}_k = \bmb{\bmb{\frac{-v}{s}}^{\frac1{k+1}}m, m}$ for any $k\in\mathbb{Z}^+$,
while for those with only the pseudoscalar self-interaction (\ie $s=0, p\neq 0,v=0$),
the feasible set becomes $\mathcal{F}_k = \varnothing$ holds for all $k\in\mathbb{Z}^+$.
\end{Remark}

\subsection{$k\in\mathbb{Z}^+$ and $\omega=m>0$}
\label{sec:w=m>0}

This subsection concerns the situation with $k\in\mathbb{Z}^+$ and $\omega=m> 0$
in which $\theta(x)$ is given in Eq.~\eqref{eq:w=m}.
Consequently, we have
\begin{align}
\cos 2\theta &= \frac{(2kmx)^2-1}{(2kmx)^2+1}\in[-1,1), \label{eq:cos-m=w} \\
\sin 2\theta &= \frac{4kmx}{(2kmx)^2+1}\in[-1,1], \nonumber \\
G(x) &= s \bmb{\frac{(2kmx)^2-1}{(2kmx)^2+1}}^{k+1}
   + {p} \bmb{\frac{4kmx }{(2kmx)^2+1}}^{k+1}+ {v}, \nonumber \\
m \cos 2\theta -\omega &= -\frac{2m}{(2kmx)^2+1} \in [-2m, 0), \label{eq:mcos2-w-m=w}
\end{align}
and the charge density becomes
\begin{equation}\label{eq:R:w=m}
j_0(x) = \left(-\frac{2m}{(1+(2kmx)^2)G(x)} \right)^{\frac{1}{k}},
\end{equation}
with $|G(x)| < |s|+|{p}|+|{v}|$.
That ${j}_0(x)>0$ shown in Eq.~\eqref{j0R2} gives
\begin{equation} \label{eq:j0>0-m=w}
G(x) < 0, \quad \forall x \in (-\infty,+\infty).
\end{equation}
Since $j_0(x) \propto \frac{1}{x^{2/k}} (x \rightarrow \infty)$,
the finite charge condition
requires $0<k <2$.
It is worth nothing that
$j_0(x)$ decays polynomially to zero as $x\rightarrow \infty$,
which is totally different from the exponential decay
happens in Section \ref{sec:m>w>=0}.
Therefore, we only need to consider the case of $k=1$
in which the restriction \eqref{eq:j0>0-m=w} becomes:
$\forall x\in(-\infty,\infty), G(x) = (s-{p}) \left(\frac{(2mx)^2-1}{(2mx)^2+1}\right)^{2} + {v}+{p} <0$,
that is equivalent to requiring
$G(0)=s + {v}< 0$ as well as  $G\big(\frac1{2m}\big) = {v} + {p} <0$
for $G(x)$ is even and has
three local extreme points $x=0,\pm\frac1{2m}$,
and ${\displaystyle \lim_{x\rightarrow\infty}}G(x) = s + {v} = G(0)$.
Accordingly,
we have
\begin{equation}
Q = \frac{\pi}{\sqrt{(s+{v})({v}+ {p})}},
\end{equation}
\ie the charge is finite.
Hence for $\omega=m>0$ and $k=1$,
we have the NLD solitary waves displayed in Eqs.~\eqref{localsol1} and \eqref{phaseexp}
satisfies the finite charge condition
if the linear combination coefficients $(s,p,v)$ belong to
\begin{equation*}
\mathcal{E}^-_1 :=
\{(s,p,v)| s + v < 0, v + p < 0\}\subset \mathcal{E}_1,
\end{equation*}
otherwise, the charge corresponding to the NLD spinor given in Eqs.~\eqref{localsol1} and \eqref{phaseexp}
can not be finite or Eq.~\eqref{j0R2} can not hold for all $x\in\mathbb{R}$.

\begin{Remark}\rm
It was pointed out that the profile of the charge density
$\rho_Q(x)$ given in Eq.~\eqref{R2} is either one-hump or two-hump
under only the quadric scalar self-interaction (\ie $k=1$, $v=p=0$)
\cite{ShaoTang2005}, which is still true for the scalar
self-interaction with more general integer exponent power (\ie
$k>1$, $v=p=0$) \cite{CooperKhareMihailaSaxena2010}. What role does
such multi-hump structure play in the interaction dynamics for the
NLD solitary waves attracts a lot of attention. Numerical results
have shown that the two-hump NLD solitary waves may collapse ({\ie
they after collision stop being solitary waves}) during the
scattering, whereas the collapse phenomena cannot be generally
observed in collisions of the one-hump NLD solitary waves
\cite{ShaoTang2005,ShaoTang2008}. {Since the collision can be
regarded as a solution of the time-containing equation with the
initial condition formed by two or more solitary waves separated
from each other by large distances, so as to be independent, we
guess, the ``instability'' is related to such collapse.
More efforts are still needed in exploring the physical mechanism of the collapse,}
such as when and why the NLD solitary waves may collapse during
their interaction dynamics.
An more direct question is naturally raised:
Is there a connection between the instability (\ie collapse) and the multi-hump structure?
Very recently,
we have further shown that collapse happens after binary collision of one-hump NLD solitary waves
under the cubic self-interaction in contrast to no collapse scattering
for corresponding quadric case \cite{XuShaoTang2013}.
In summary, both the multi-hump (two-hump) structure and
high order nonlinearity could undermine
the stability during the scattering of the NLD solitary waves.
In the next section,
we will show that
the multi-hump structure depends on
the linear combination coefficients $(s,p,v)\in\mathcal{E}_k$ and
the integer power exponent $k+1$.
\end{Remark}

\section{Multi-hump structure}
\label{sec:multi-hump}

This section will focus on investigating
the multi-hump structure of
the NLD solitary wave \eqref{localsol1} with
the linear combined self-interaction \eqref{generalLI}
for $(s,p,v)\in\mathcal{E}_k$ and $\omega\in\mathcal{F}_k$ when $k\in\mathbb{Z}^+$ and $m>\omega\geq 0$
as well as
for $(s,p,v)\in\mathcal{E}_1^-$ when $k=1$ and $m=\omega > 0$.
More specifically,
we will answer:
Can the charge density $\rho_Q(x)$ have more humps than two under the linear combined self-interaction \eqref{generalLI}?
At most how many humps can the charge density $\rho_Q(x)$ afford?
Can we have similar results for the energy density or the momentum density?

\subsection{$k\in\mathbb{Z}^+$ and $m>\omega\geq 0$}
\label{sec:hump:m>w>=0}

The number of humps in the charge density $\rho_Q(x)$ is equal to
the number of its local maximum and can be determined by the number of zeros of
$\dd{\rho_Q(x)}{x}$,
and the zeros of $\dd{\rho_Q(x)}{x}$ is the same as those of $\dd{\rho_Q^k(x)}{x}$
for $\rho_Q(x)> 0$ and $k\in\mathbb{Z}^+$.
When $m>\omega\geq0$,
for convenience,
we introduce an intermediate variable $\xi=2\theta$ and
rewrite $\rho_Q(x)$ in terms of $\xi$ into
\begin{equation}\label{j-xi}
\hat{\rho}_Q(\xi) = m ^\frac1k \bmb{\frac{\cos\xi - a}{\hat{G}(\xi)}}^{\frac1k}>0,
\quad \forall\xi \in I,
\end{equation}
where
\begin{align*}
a:&=\frac{\omega}{m}\in[0,1), \\
\hat{G}(\xi) :&= s \cos^{k+1}\xi + p \sin^{k+1}\xi + v, \\
I:&=(-\cos^{-1}(a), \cos^{-1}(a))\subset
\bmb{-\frac{\pi}2,\frac{\pi}2}.
\end{align*}
Combining Eqs.~\eqref{theta-recast1}, \eqref{eq:mcos2-w-m>w} and \eqref{eq:relation-z-y}
gives that $\dd{\xi}{x} >0$ hold for all $x\in\mathbb{R}$,
and then the chain rule
\begin{equation}\label{eq:j0-k}
\dd{\rho_Q^k(x)}{x} = \dd{\hat{\rho}_Q^k(\xi)}{\xi} \dd{\xi}{x}=-\hat{\rho}_Q^{2k}(\xi)\dd{\hat{\rho}_Q^{-k}(\xi)}{\xi} \dd{\xi}{x}
\end{equation}
further implies that
$\dd{\rho_Q^k(x)}{x}$ has the same zeros as $\dd{\hat{\rho}_Q^{-k}(\xi)}{\xi}$.
That is, the remaining task is to determine or estimate the number of zeros of
$\dd{\hat{j}_0^{-k}(\xi)}{\xi}$. To this end,
technically, we need the following two lemmas
in which $\numzero_\Omega[f]$ (resp. $\numextr_\Omega[f]$) represents
the number of zeros (resp. extreme points at which the derivatives of $f(\xi)$ are zero) of function $f(\xi)\in C^1(\Omega)$
in an open interval $\Omega$. 

\begin{Lemma}\rm\label{lemma1}
Given $f(\xi)\in C^1(\Omega)$, we have

(i) $\numzero_\Omega[f]\leq \numextr_\Omega[f] + 1$;

(ii) $\numzero_\Omega[\alpha f] = \numzero_\Omega[f]$ and $\numextr_\Omega[\alpha f]=\numextr_\Omega[f]$
hold for any $\alpha\neq 0$.
\end{Lemma}

\begin{Lemma}\label{le:fOverg}\rm
Suppose $f(\xi),g(\xi),\frac{g^\prime(\xi)}{f^\prime(\xi)}\in C^1(\Omega)$,
and $f(\xi)\neq 0$ holds for all $\xi \in \Omega$.
Then $$\numextr_\Omega\bbmb{\frac{g}{f}} \leq \numextr_\Omega\bbmb{\frac{g'}{f'}} + \numzero_\Omega[f']+1.$$
\end{Lemma}
\begin{proof}
Since $\bmb{\frac{g'}{f'} f - g}' = \bmb{\frac{g'}{f'}}'f$,
then
\begin{equation}\label{lemma2-eq1}
\numextr_\Omega\bbmb{\frac{g'}{f'} f - g} = \numextr_\Omega\bbmb{\frac{g'}{f'}},
\end{equation}
for $f(\xi)\neq 0$ holds for all $\xi \in \Omega$.
Similarly, since
$\bmb{\frac{g}{f}}' = \frac{f'}{f^2}\bmb{\frac{g'}{f'} f - g}$,
then
\begin{align*}
\numextr_\Omega \bbmb{\frac{g}{f}}
&\leq \numzero_\Omega\bbmb{\frac{g'}{f'} f - g} + \numzero_\Omega[f'] \\
&\leq \numextr_\Omega\bbmb{\frac{g'}{f'} f - g} + 1 + \numzero_\Omega[f']\\
&\leq \numextr_\Omega\bbmb{\frac{g'}{f'}} + \numzero_\Omega[f']+ 1,
\end{align*}
where we have used Lemma \ref{lemma1}(i) in the second line
and Eq.~\eqref{lemma2-eq1} in the last line.
\end{proof}

We are now in position to determine
$\numextr_I\bbmb{\hat{\rho}_Q^{-k}}$, \ie
the number of extreme points
of $\hat{\rho}_Q^{-k}(\xi)$ in
the interval $I$.
Because
\begin{equation}\label{j0-kxi}
   \hat{\rho}_Q^{-k}(\xi) = \frac{1}{m} \frac{\hat{G}(\xi)}{\cos\xi - a} =: \frac{1}{m} \frac{g_1(\xi)}{f_1(\xi)},
\end{equation}
for Eq.~\eqref{j-xi}, using Lemma \ref{lemma1}(ii) and Lemma \ref{le:fOverg} directly gives
\begin{equation}\label{Eg1/f1}
\numextr_I\bbmb{\hat{\rho}_Q^{-k}} = \numextr_I\bbmb{\frac{1}{m}\frac{g_1}{f_1}} = \numextr_I\bbmb{\frac{g_1}{f_1}}
\leq \numextr_I\bbmb{\frac{g_1^\prime}{f_1^\prime}}+\numzero_I[f_1^\prime]+1.
\end{equation}
Direct calculation shows
\begin{align}
g_1'(\xi) &= (k+1)(-s \cos^k\xi \sin \xi + p \sin^k\xi\cos\xi), \label{g1p}\\
f_1'(\xi) &= -\sin \xi, \nonumber
\end{align}
thus
\begin{equation}\label{Zf1}
\numzero_I[f_1^\prime] = 1,
\end{equation}
and
\begin{equation}\label{g1p/f1p}
\frac{g_1'(\xi)}{f_1'(\xi)}
= \frac{(k+1)(s - p \tan^{k-1}\xi)}{\cos^{-k}\xi} = :
    (k+1)\frac{g_2(\xi)}{f_2(\xi)}
\end{equation}
which implies by Lemma \ref{lemma1}(ii) and Lemma \ref{le:fOverg} that
\begin{equation}\label{Eg1p/f1p}
\numextr_I\bbmb{\frac{g_1^\prime}{f_1^\prime}}=\numextr_I\bbmb{(k+1)\frac{g_2}{f_2}}
=\numextr_I\bbmb{\frac{g_2}{f_2}}\leq \numextr_I\bbmb{\frac{g_2^\prime}{f_2^\prime}}
+\numzero_I[f_2^\prime]+1,
\end{equation}
for $k\in\mathbb{Z}^+$.
It can easily be shown that
\begin{align}
    g_2'(\xi) &= - p(k-1) \tan^{k-2}\xi \cos^{-2}\xi, \label{g2p}  \\
    f_2'(\xi) &= k \cos^{-k-1}\xi \sin\xi, \nonumber
\end{align}
and then
\begin{align}
\numzero_I[f_2^\prime] &= 1, \label{Zf2}\\
\frac{g_2'(\xi)}{f_2'(\xi)}  &=-\tfrac{p(k-1)}{k}\sin^{k-3}\xi \cos\xi, \label{g2p/f2p}\\
\bmb{\frac{g_2'(\xi)}{f_2'(\xi)}}^\prime  &
= - \tfrac{p(k-1)}{k} \sin^{k-4}\xi\bmb{(k-3)\cos^2\xi -\sin^2\xi}. \label{Eg2p/f2p}
\end{align}

Based on the foregoing derivation,
we can rigorously determine the number of humps of the charge density $\rho_Q(x)$
and the results are summarized as follows.

\noindent $\bullet$ \textbf{Case Q1}  When $k\in\mathbb{Z}^+$ and $s=p=0$, we have $v\neq 0$ for $(s,p,v)\in\mathcal{E}_k$
and Eq.~\eqref{j0-kxi} becomes
\begin{equation*}
\hat{\rho}_Q^{-k}(\xi) = \frac{v}{m} \frac{1}{\cos\xi - a},
\end{equation*}
from which it can readily be seen that the charge density has only one hump.

\noindent $\bullet$ \textbf{Case Q2}  When $k\in\mathbb{Z}^+$, $p=0$ and $s\neq 0$,
Eq.~\eqref{g1p/f1p} becomes
\begin{equation*}
\frac{g_1'(\xi)}{f_1'(\xi)}
= (k+1)s \cos^k\xi,
\end{equation*}
then we have $\numextr_I\bbmb{\frac{g_1'}{f_1'}}=1$,
thus $\numextr_I\bbmb{\hat{\rho}_Q^{-k}} \leq 3$ for Eq.~\eqref{Eg1/f1},
\ie the charge density  has at most two humps.

\noindent $\bullet$ \textbf{Case Q3}  When $k=1$ and $s = p\neq 0$,
Eq.~\eqref{j0-kxi} becomes
\begin{equation*}
\hat{\rho}_Q^{-1}(\xi) = \frac{s+v}{m} \frac{1}{\cos\xi - a},
\end{equation*}
from which it can readily be seen that the charge density has only one hump.

\noindent $\bullet$ \textbf{Case Q4}  When $k=1$, $p\neq 0$ and $s \neq p$,
Eq.~\eqref{g1p/f1p} becomes
\begin{equation*}
\frac{g_1'(\xi)}{f_1'(\xi)}
= 2(s-p) \cos\xi,
\end{equation*}
and we have $\numextr_I\bbmb{\frac{g_1'}{f_1'}}=1$,
thus $\numextr_I\bbmb{\hat{\rho}_Q^{-1}} \leq 3$ for Eq.~\eqref{Eg1/f1},
\ie the charge density $\rho_Q(x)$ has at most two humps.

\noindent $\bullet$ \textbf{Case Q5}  When $k=2$ and $p\neq 0$, Eq.~\eqref{g1p/f1p} becomes
\begin{equation*}
  \frac{g_1'(\xi)}{f_1'(\xi)} = 3(s \cos^2\xi  - p \sin\xi\cos\xi),
\end{equation*}
then we have
 \begin{equation*}
   \bmb{\frac{g_1'(\xi)}{f_1'(\xi)}}^\prime = -3(s\sin(2\xi) + p\cos(2\xi)),
 \end{equation*}
 and $\numextr_I\bbmb{\frac{g_1'}{f_1'}}\leq 2$,
 thus $\numextr_I\bbmb{\hat{\rho}_Q^{-2}} \leq 4$ for Eq.~\eqref{Eg1/f1},
\ie the charge density $\rho_Q(x)$ has at most two humps.

\noindent $\bullet$ \textbf{Case Q6}  When $k=3$ and $p\neq 0$, Eq.~\eqref{g2p/f2p} becomes
\begin{equation*}
\frac{g_2'(\xi)}{f_2'(\xi)}  =-\frac{2p}{3} \cos\xi,
\end{equation*}
then we have $\numextr_I\bbmb{\frac{g_2'}{f_2'}}=1$,
thus $\numextr_I\bbmb{\hat{\rho}_Q^{-3}} \leq 5$ for Eqs.~\eqref{Eg1/f1}, \eqref{Zf1}, \eqref{Eg1p/f1p}
and \eqref{Zf2},
\ie the charge density  has at most three humps.

\noindent $\bullet$ \textbf{Case Q7}  When $k=4$ and $p\neq 0$, Eq.~\eqref{Eg2p/f2p} becomes
\begin{equation*}
  \bmb{\frac{g_2'(\xi)}{f_2'(\xi)}}^\prime
= - \frac{3p}{4} \cos (2\xi),
\end{equation*}
and implies that
$\frac{g_2'(\xi)}{f_2'(\xi)}$ has at most two extreme points,
then we have
$
\numextr_I\bbmb{\frac{g_2^\prime}{f_2^\prime}}\leq 2,
$
thus
$
\numextr_I\bbmb{\hat{\rho}_Q^{-4}} \leq 6,
$
for Eqs.~\eqref{Eg1/f1}, \eqref{Zf1}, \eqref{Eg1p/f1p}
and \eqref{Zf2},
which means that the charge density has at most three humps.

\noindent $\bullet$ \textbf{Case Q8}  When $k\in\{5,6,7,\cdots\}$ and $p\neq 0$,
Eq.~\eqref{Eg2p/f2p} implies that
the extreme points of $\frac{g_2'(\xi)}{f_2'(\xi)}$ possibly locate at $\xi=0$ and $\xi=\pm\tan^{-1}(\sqrt{k-3})$
and thus
\begin{equation}\label{Eg2p/f2p-1}
\numextr_I\bbmb{\frac{g_2^\prime}{f_2^\prime}}\leq 3.
\end{equation}
Combining Eqs.~\eqref{Eg1/f1}, \eqref{Zf1}, \eqref{Eg1p/f1p},
\eqref{Zf2} and \eqref{Eg2p/f2p-1} leads to
\begin{equation}
\numextr_I\bbmb{\hat{\rho}_Q^{-k}} \leq 7,
\end{equation}
which means that the charge density has at most four humps.

\begin{Remark}\rm
Our analysis has shown that:
(i) the charge density has only one hump under
the pure vector self-interaction \cite{CooperKhareMihailaSaxena2010}
and has either one hump or two humps under
the linear combination of the scalar and vector self-interactions \cite{XuShaoTang2013};
(ii)
The charge density has at most four humps
for $(s,p,v)\in\mathcal{E}_k$ and $\omega\in\mathcal{F}_k$ when $k\in\mathbb{Z}^+$ and $m>\omega\geq 0$;
(iii)
The NLD solitary wave with the four-hump charge density
can only appear
in the situation of higher nonlinearity, \ie $k\in\{5,6,7,\cdots\}$,
while for the case of $k\in\{1,2\}$ (resp. $k\in\{3,4\}$),
the charge density has at most two (resp. three) humps;
(iv) under the linear combination of the vector and pseudoscalar self-interactions
(\ie $v\neq 0,p\neq 0,s=0$) with $k\in\{3,4,5\cdots\}$
the charge density also has at most three humps
because setting $s=0$ in Eq.~\eqref{g1p/f1p} leads to
\[
\bmb{\frac{g_1'(\xi)}{f_1'(\xi)}}^\prime =  ((k-1)\cot^2\xi-1) \sin^{k}\xi,
\]
which has at most three zeros,
\ie $\numextr_I\bbmb{\frac{g_1^\prime}{f_1^\prime}}\leq 3$,
then $\numextr_I\bbmb{\hat{\rho}_Q^{-k}} \leq 5$ for Eqs.~\eqref{Eg1/f1} and \eqref{Zf1};
(v) The charge density can indeed have three humps or four humps as shown
in Figs.~\ref{three_hump} and \ref{four_hump},
while the two-hump charge density was first pointed out in \cite{ShaoTang2005}.
\end{Remark}

\begin{figure}[ht]
  \centering
\subfigure[~The charge density $\rho_Q(x)$.]{
\includegraphics[width=8cm, height=6cm]{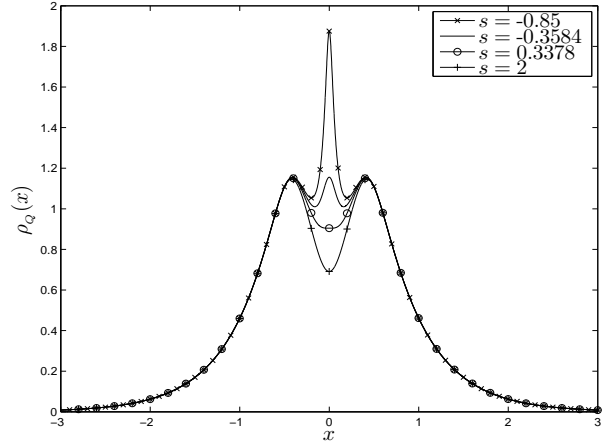}}
\subfigure[~The energy density $\rho_E(x)$.]{
\includegraphics[width=8cm, height=6cm]{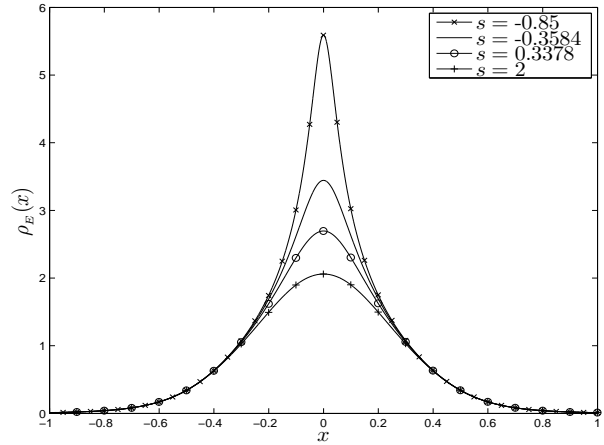}}
\caption{\small
  The two-hump and three-hump charge densities are plotted in (a) with respect to $s$
  with other parameters being $m=1,k=3,\omega=0.01,{v}=1, {p}=-0.95$.
  The critical value for the two-hump charge densities transiting to the three-hump ones is $s=\frac{25}{74}\simeq 0.3378$.
  When $s\simeq-0.3584$,
  the three peaks have the same height $1.156$ and locate at $x=0$ and $x \simeq \pm 0.4110$.
  It is noted that the energy densities with the same parameters have just one hump, see (b).
}
\label{three_hump}
\end{figure}

\begin{figure}[ht]
  \centering
\subfigure[~The charge density $\rho_Q(x)$.]{
\includegraphics[width=8cm, height=6cm]{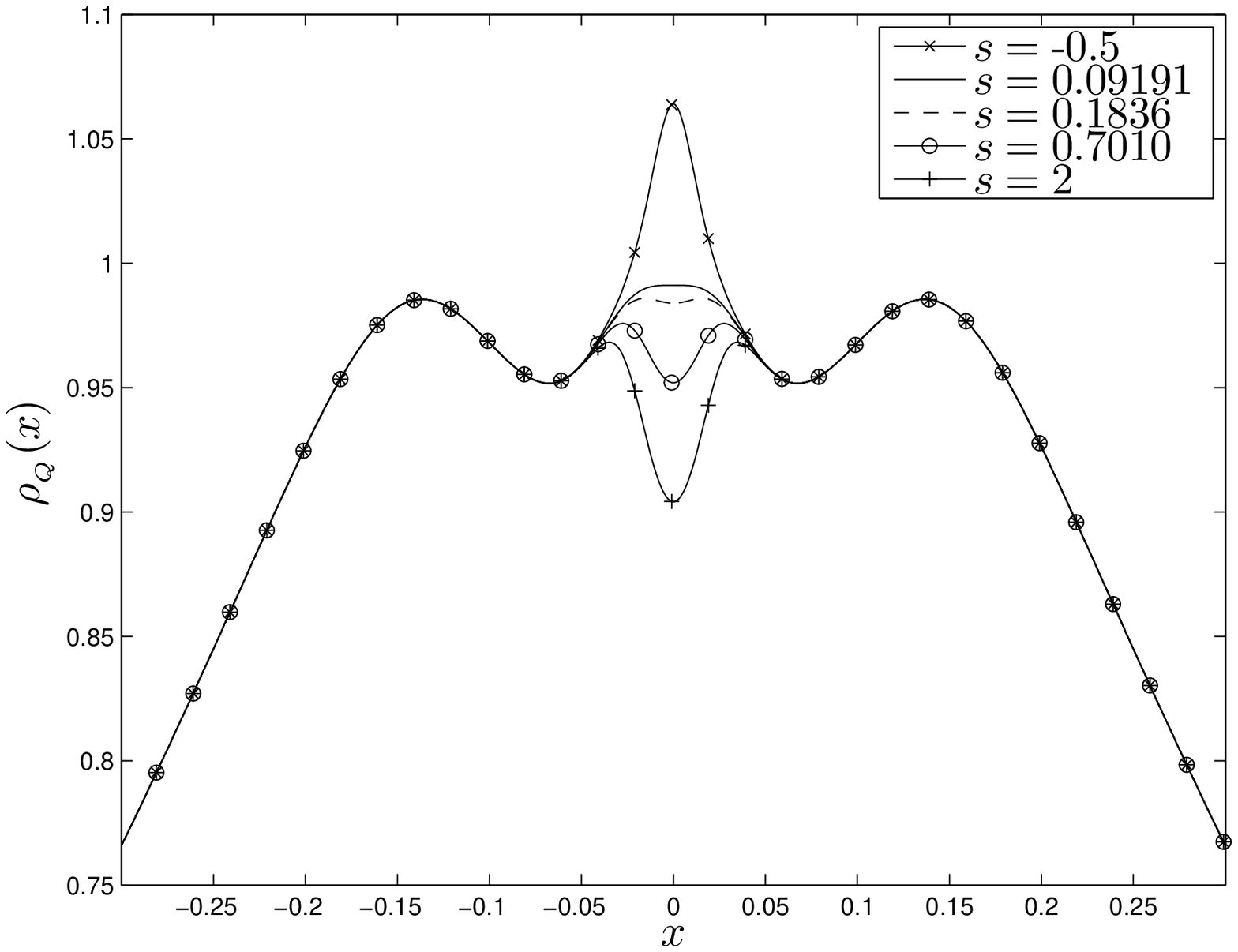}}
\subfigure[~The energy density $\rho_E(x)$.]{
\includegraphics[width=8cm, height=6cm]{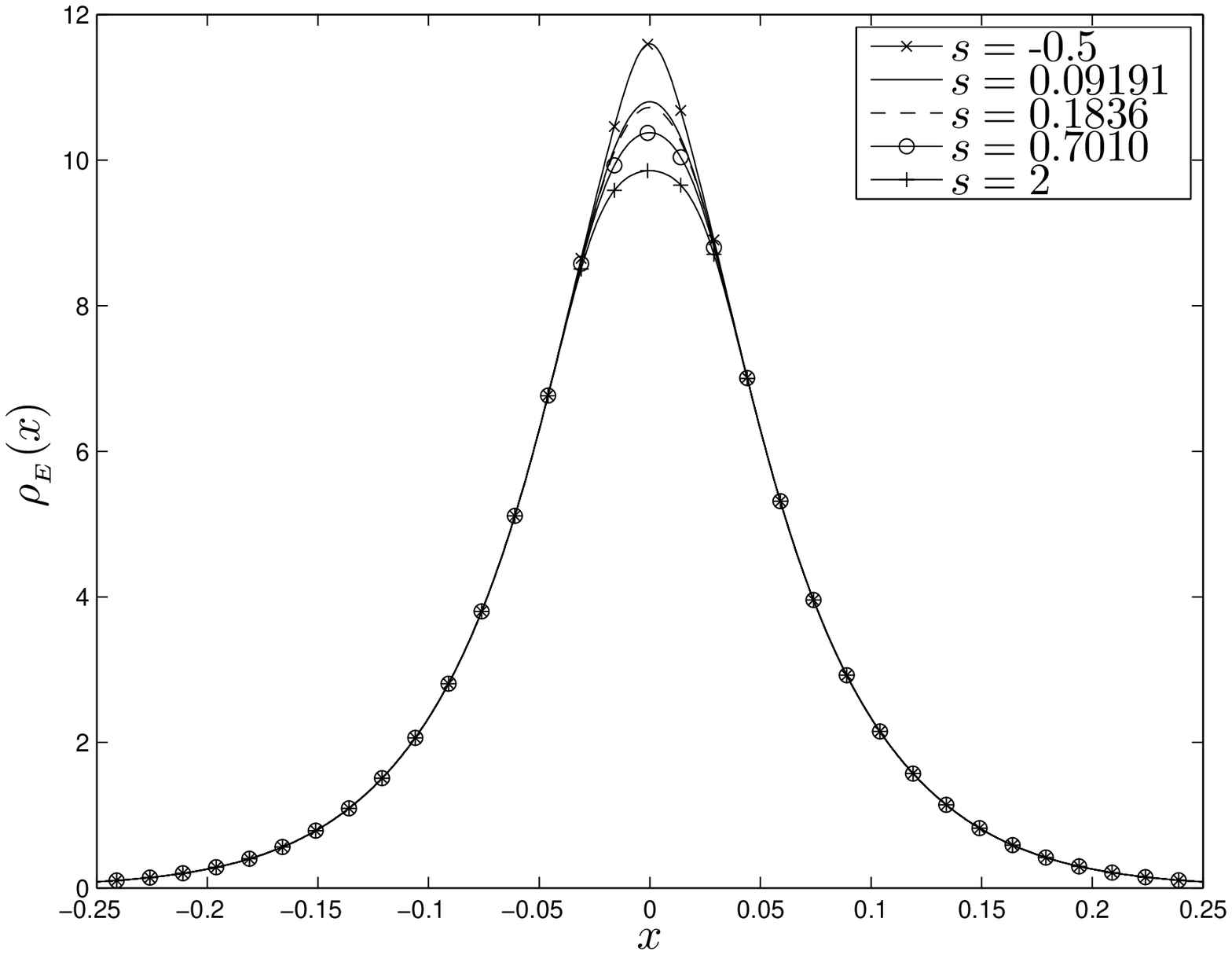}}
\caption{\small
The three-hump and four-hump charge densities are plotted in (a) with respect to $s$
  with other parameters being $m=1,k=11,\omega=0.01,{v}=1, {p}=-0.95$.
  The critical value for the three-hump charge densities transiting
  to the four-hump ones is $s=\frac{25}{272}\simeq 0.09191$.
  The four peaks have the same hight for $s\simeq 0.1836$ and
  so do the three valleys for $s\simeq 0.7010$.
  It is noted that the energy densities with the same parameters have only one hump, see (b).}
\label{four_hump}
\end{figure}

Apart from the charge $Q$ in \eqref{Q-2C},
there are another two important conservative quantities:
the energy $E$ in \eqref{E-2C} and the momentum $P$ in \eqref{P-2C}.
For the solitary wave solutions with the form in Eqs.~\eqref{localsol1} and \eqref{phaseexp},
from Eq.~\eqref{T01=0},
the momentum density $\rho_P(x):=T^{01}[\vPsi](x,t)$ in Eq.~\eqref{t_01-sw} vanishes for all $x\in\mathbb{R}$,
which reflects that the NLD solitary waves are at rest (\ie the standing waves),
while the energy density $\rho_E(x) := T^{00}[\vPsi](x,t)$ in \eqref{t_00-sw} becomes
\begin{equation}\label{T00-x}
  \rho_E(x) = \rho_Q(x)(m k\cos(2\theta(x)) -(k-1)\omega)>0, \quad \forall x\in\mathbb{R},
\end{equation}
for Eqs.~\eqref{key-Li}, \eqref{key-eq1}, \eqref{localsol1} and \eqref{phaseexp}.
Next,
we are going to investigate the multi-hump structure of the energy density $\rho_E(x)$
and the method is similar to that used in discussing the multi-hump structure of
the charge density $\rho_Q(x)$.

Rewrite $\rho_E(x)$ in terms of the intermediate variable $\xi$ into
\begin{equation}
  \label{eq:t00-xi}
  \hat{\rho}_E(\xi) = mk \hat{\rho}_Q(\xi) (\cos\xi -b)>0,\quad \forall\xi\in I,
\end{equation}
where $b = \frac{k-1}k a\leq a$
and Eq.~\eqref{j-xi} is applied,
and the number of extreme points
of $\hat{\rho}_E^{-k}(\xi)$ in
the interval $I$, \ie $\numextr_I\bbmb{\hat{\rho}_E^{-k}}$,
is to be estimated. Because
  \begin{equation}\label{T00-kxi}
    \hat{\rho}_E^{-k}(\xi) =
{\frac1{m^{k+1}k^k}}\frac{\hat{G}(\xi)}{(\cos\xi - a)(\cos\xi - b)^k}
={\frac1{m^{k+1}k^k}}\frac{g_1(\xi)}{h_1(\xi)},
  \end{equation}
for Eqs.~\eqref{j0-kxi} and \eqref{eq:t00-xi},
where $h_1(\xi) := (\cos\xi - a)(\cos\xi - b)^k$,
using Lemma \ref{lemma1}(ii) and Lemma \ref{le:fOverg} gives
\begin{equation}\label{Eg4/f4}
\numextr_I\bbmb{\hat{\rho}_E^{-k}} = \numextr_I\bbmb{{\frac1{m^{k+1}k^k}} \frac{g_1}{h_1}} = \numextr_I\bbmb{\frac{g_1}{h_1}}
\leq \numextr_I \bbmb{\frac{g_1^\prime}{h_1^\prime}}+\numzero_I[h_1^\prime]+1.
\end{equation}
Direct calculation shows
\begin{equation}\label{h1p}
    h_1'(\xi)  =-\sin\xi (\cos\xi-b)^{k-1}(k+1)(\cos\xi - c),
\end{equation}
where $c=\tfrac{ka+b}{k+1}=\tfrac{k(k+1)-1}{k(k+1)}a$
and $\cos\xi - b \geq \cos\xi - c \geq \cos\xi- a>0$ holds for all $\xi\in I$,
thus
\begin{equation}\label{Zf4}
\numzero_I[h_1^\prime] = 1.
\end{equation}
Combining Eqs.~\eqref{g1p}, \eqref{g1p/f1p} and \eqref{h1p} leads to
\begin{equation}\label{g1p/h1p}
    \frac{g_1'(\xi)}{h_1'(\xi)} = \frac{s - p \tan^{k-1}\xi}{(1- \tfrac{b}
    {\cos\xi})^{k-1}(1 - \tfrac{c}{\cos\xi})}  = \frac{g_2(\xi)}{h_2(\xi)},
\end{equation}
where $h_2(\xi):={(1- \tfrac{b}{\cos\xi})^{k-1}(1 - \tfrac{c}{\cos\xi})}$,
and implies by Lemma \ref{le:fOverg} that
\begin{equation}\label{Eg4p/f4p}
\numextr_I\bbmb{\frac{g_1^\prime}{h_1^\prime}}
=\numextr_I\bbmb{\frac{g_2}{h_2}}\leq \numextr_I\bbmb{\frac{g_2^\prime}{h_2^\prime}}
+\numzero_I[h_2^\prime]+1.
\end{equation}
It can easily be shown that
  \begin{equation} \label{h2p}
    h_2'(\xi) = -(1- \tfrac{b}{\cos\xi})^{k-2} \tfrac{\sin\xi}{\cos^2\xi}
    \bmb{b(k-1)(1 - \tfrac{c}{\cos\xi})+ c(1-\tfrac{b}{\cos\xi}) },
  \end{equation}
thus we have
\begin{equation}
\numzero_I[h_2^\prime] = 1, \quad \text{if}\;\; b\neq 0, \label{Zf5}
\end{equation}
and
\begin{equation}\label{g2p/h2p}
\frac{g_2'(\xi)}{h_2'(\xi)}
     = p \frac{(k-1)\sin^{k-3}\xi} {({\cos\xi}-b)^{k-2}
    \bmb{b(k-1)(1 - \tfrac{c}
        {\cos\xi})+ c(1-\tfrac{b}{\cos\xi})}},
\end{equation}
for Eq.~\eqref{g2p}. Hence we are able to determine the number of humps of the energy density $\rho_E(x)$
and the results are shown below.

\noindent $\bullet$ \textbf{Case E1} When $k\in\mathbb{Z}^+$ and $p=\omega=0$,
Eq.~\eqref{T00-kxi} becomes
  \begin{align*}
    \hat{\rho}_E^{-k}(\xi) = \frac1{m^{k+1}k^k}
    \bmb{s+\frac{v}{\cos^{k+1}\xi}},
  \end{align*}
from which it can readily be seen that the energy density has only one hump.

\noindent $\bullet$ \textbf{Case E2}  When $k\in\mathbb{Z}^+$, $p=0$ and $\omega\neq 0$,
we have $0<\omega<m$, $a>0$, $b>0$, and Eq.~\eqref{T00-kxi} becomes
  \begin{equation}\label{T00-kxi-1}
    \hat{\rho}_E^{-k}(\xi)
    = \frac1{m^{k+1}k^k}\frac{q_1(\xi)}{q_2(\xi)},
  \end{equation}
  where $q_1(\xi):={s + \tfrac{v}{\cos^{k+1}\xi}}$ and $q_2(\xi):={(1 -
      \tfrac{a}{\cos\xi})(1 -\frac{b}{\cos\xi})^k}$.
  It is easy to show that both $q_1(\xi)$ and $q_2(\xi)$ are even and positive in
  the domain $I$. In fact, we can further show that
  $\frac{q_1(\xi)}{q_2(\xi)}$ increases monotonously as $\xi$ goes from $0$ to $\cos^{-1}a$
  which implies that the energy density has only one hump in this situation.
  The reason is given in the following.
  The case of $v\geq 0$ is trivial. If $v<0$,
  then we have $s>0$ for $\hat{G}(0)=s + v>0$ because $\omega\in\mathcal{F}_k$,
  and using the formula for difference of powers leads to
  \begin{equation}\label{q1}
  q_1(\xi)= \bmb{\rho -\frac{\eta}{\cos\xi}}
      \sum_{j=0}^{k}\rho ^j \bmb{\frac{\eta}{\cos\xi}}^{k-j},
  \end{equation}
  where $\rho := s^{\frac1{k+1}}> 0$ and $\eta := (-v)^{\frac1{k+1}}>0$.
  It is simple to see that ${\rho -\frac{\eta}{\cos\xi}}>0$ holds for all $x\in [0,\cos^{-1}a)$
  and the limit gives ${a} \rho-{\eta} \geq0$ when $\xi\rightarrow \cos^{-1}a$.
  Combining Eqs.~\eqref{T00-kxi-1} and \eqref{q1} yields
  \begin{equation}\label{q1/q2}
      \frac{q_1(\xi)}{q_2(\xi)}
      = \bmb{\frac{a \rho - \eta}{a(1 - \tfrac{a}{\cos\xi})(1
          -\frac{b}{\cos\xi})^k} + \frac{\eta}{a(1
          -\frac{b}{\cos\xi})^k} }\sum_{j=0}^{k}\rho ^j
      \bmb{\frac{\eta}{\cos\xi}}^{k-j},
    \end{equation}
  where the identity $\rho -\frac{\eta}{\cos\xi} =
    \rho - \frac{\eta}{a} + \frac{\eta}{a}\bmb{1
      -\frac{a}{\cos\xi}}$ is applied.
  From Eq.~\eqref{q1/q2}, it is trivial to see that
  $\frac{q_1(\xi)}{q_2(\xi)}$ increases monotonously in the domain $[0,\cos^{-1}a)$.

\noindent $\bullet$ \textbf{Case E3}  When $k=1$ and $p\neq 0$,
we have
$$
\hat{\rho}_E^{-1} ={\frac1{m^{2}}}\frac{ (s-p) + \tfrac{v+p}{\cos^{2}\xi}}{(1 - \tfrac{a}{\cos\xi})(1 - \tfrac{b}{\cos\xi})}
$$
and thus the energy density also has only one hump
by utilizing an argument similar to that used in the situation of $k\in\mathbb{Z}^+$ and $p\neq 0$
(see above \textbf{Case E1} and \textbf{Case E2}).

\noindent $\bullet$ \textbf{Case E4}  When $k\in\{2,3,4,\cdots\}$, $p\neq 0$ and $\omega=0$,
we have $a=b=0$, and Eq.~\eqref{g1p/h1p} becomes
\begin{equation*}
 \frac{g_1'(\xi)}{h_1'(\xi)}= s - p \tan^{k-1}\xi,
\end{equation*}
which implies that
$\numextr_I\bbmb{\frac{g_1'}{h_1'}} \leq 1$
and thus $\numextr_I\bbmb{\hat{\rho}_E^{-k}} \leq 3$ for Eqs.~\eqref{Eg4/f4} and \eqref{Zf4}.
That is, the energy density  has at most two humps.

\noindent $\bullet$ \textbf{Case E5}  When $k=2$, $p\neq 0$ and $\omega\neq 0$,
Eq.~\eqref{g1p/h1p} becomes
\begin{equation}
\frac{1}{p} \frac{g_2(\xi)}{h_2(\xi)} = \frac{\tfrac{s}{p} -  \tan\xi}{(1- \tfrac{b}
    {\cos\xi})(1 - \tfrac{c}{\cos\xi})},
\end{equation}
from which it is easy to see that $\frac{1}{p} \frac{g_2(\xi)}{h_2(\xi)}$
decreases monotonously as $\xi$ goes from
$-\cos^{-1}a$ to $\cos^{-1}a$ when $s=0$
or as $\xi$ goes from
$-\cos^{-1}a$ to $0$ when $\frac{s}{p}>0$,
and $\xi=0$ is not the extreme point of $\frac{1}{p} \frac{g_2(\xi)}{h_2(\xi)}$
for $\left(\frac{1}{p} \frac{g_2(\xi)}{h_2(\xi)}\right)^\prime_{\xi=0}=-\frac{1}{(1-b)(1-c)}<0$.
That is, $\numextr_I\bbmb{\frac{g_2}{h_2}}=0$ holds for $s=0$
and
$\numextr_I\bbmb{\frac{g_2}{h_2}}=\numextr_{I_1}\bbmb{\frac{g_2}{h_2}}$
is true for $\frac{s}{p}>0$ where $I_1=(0,\cos^{-1}a)$.
When $\frac{s}{p}>0$, using Lemma \ref{le:fOverg} further gives
\begin{equation}
\numextr_I\bbmb{\frac{g_2}{h_2}}=\numextr_{I_1}\bbmb{ \frac{g_2}{h_2}} \leq
\numextr_{I_1}\bbmb{ \frac{g_2^\prime}{h_2^\prime}} + \numzero_{I_1}\bbmb{h_2^\prime} +1 \leq 2,
\end{equation}
where in the last inequality we have used $\numzero_{I_1}\bbmb{h_2^\prime}=0$ for Eq.~\eqref{h2p} as well
as $\numextr_{I_1}\bbmb{ \frac{g_2^\prime}{h_2^\prime}}\leq 1$ for
\begin{equation*}
  \bmb{\frac{g_2'(\xi)}{h_2'(\xi)}}^\prime = \frac{p(2bc - (b+c)\cos^3\xi)}{(h_2^\prime(\xi)\cos^3\xi)^2}.
\end{equation*}
By a similar argument, we can easily show that $\numextr_I\bbmb{\frac{g_2}{h_2}}\leq 2$ also holds for
$\frac{s}{p}<0$. Therefore, $\numextr_I\bbmb{\frac{g_2}{h_2}}\leq 2$ is always true
and thus we have $\numextr_I\bbmb{\hat{\rho}_E^{-k}}\leq 4$
for Eqs.~\eqref{Eg4/f4}, \eqref{Zf4} and \eqref{g1p/h1p},
which means that the energy density has at most two humps.

\noindent $\bullet$ \textbf{Case E6}  When $k\in\{3,5,7,\cdots\}$, $p\neq 0$ and $\omega\neq 0$,
from Eq.~\eqref{g2p/h2p}, we find that
$\frac1p \frac{g_2'(\xi)}{ h_2'(\xi)}$ is even and
increases monotonously as $\xi$ goes from $0$ to $\cos^{-1}a$,
then $\numextr_I\bbmb{\frac{g_2'}{h_2'}}=1$
and $
\numextr_I\bbmb{\hat{\rho}_E^{-k}} \leq 5
$
for Eqs.~\eqref{Eg4/f4}, \eqref{Zf4}, \eqref{Eg4p/f4p}
and \eqref{Zf5},
which means that the energy density has at most three humps.

\noindent $\bullet$ \textbf{Case E7}  When $k\in\{4,6,8,\cdots\}$, $p\neq 0$ and $\omega\neq 0$,
from Eq.~\eqref{g2p/h2p}, we find that
$\frac1p \frac{g_2'(\xi)}{ h_2'(\xi)}$ is odd and
increases monotonously as $\xi$ goes from $0$ to $\cos^{-1}a$,
then $\numextr_I\bbmb{\frac{g_2'}{h_2'}} =0$
and
$
\numextr_I\bbmb{\hat{\rho}_E^{-k}} \leq 4
$
for Eqs.~\eqref{Eg4/f4}, \eqref{Zf4}, \eqref{Eg4p/f4p}
and \eqref{Zf5},
which means that the energy density has at most two humps.

\begin{Remark}\rm
Our analysis has shown that:
(i) the energy density has only one hump under
the linear combination of the scalar and vector self-interactions;
(ii)
The energy density has at most three humps
for $(s,p,v)\in\mathcal{E}_k$ and $\omega\in\mathcal{F}_k$ when $k\in\mathbb{Z}^+$ and $m>\omega\geq 0$;
(iii)
 The NLD solitary wave with the three-hump energy density can only appear
 in the situation of higher nonlinearity of  even power, \ie $k \in\{ 3,5,7,\cdots\}$,
 while for the case of $k\in\{2, 4,6,\cdots\}$,
 the energy density has at most two humps;
(iv) under the linear combination of the vector and pseudoscalar self-interactions
(\ie $v\neq 0,p\neq 0,s=0$) with $k\in\{3,5,7\cdots\}$,
the energy density also has at most two humps
because setting $s=0$ in Eq.~\eqref{g1p/h1p} leads to
\[
-\frac{1}{p}\frac{g_1'(\xi)}{h_1'(\xi)} =  \frac{\tan^{k-1}\xi}{(1- \tfrac{b}
    {\cos\xi})^{k-1}(1 - \tfrac{c}{\cos\xi})},
\]
which is even and increases monotonously as $\xi$ goes from $0$ to $\cos^{-1}a$,
then $\numextr_I\bbmb{\frac{g_1'}{h_1'}}=1$
and $
\numextr_I\bbmb{\hat{\rho}_E^{-k}} \leq 3
$
for Eqs.~\eqref{Eg4/f4} and \eqref{Zf4};
 (v) The energy density can indeed have two humps or three humps as shown
in Fig.~\ref{T00_three_hump}.
\end{Remark}

\begin{figure}[h]
  \centering
\subfigure[~The energy density $\rho_E(x)$.]{
   \includegraphics[width=8cm, height=6cm]{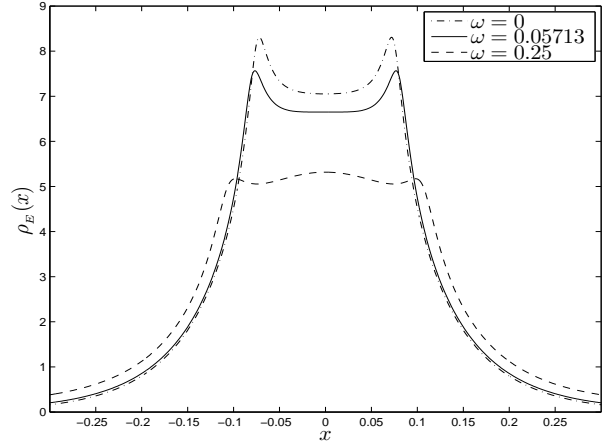}}
\subfigure[~The charge density $\rho_Q(x)$.]{
   \includegraphics[width=8cm, height=6cm]{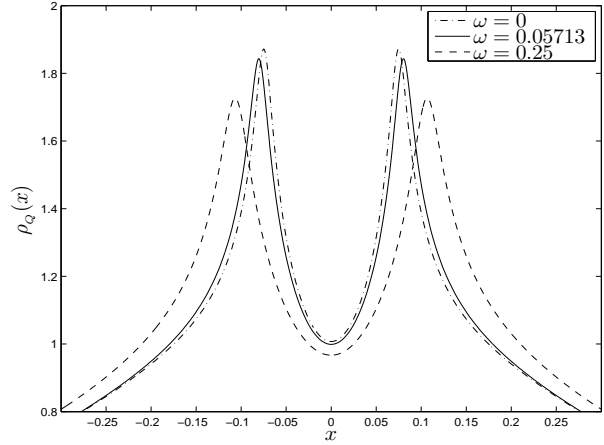}}
\caption{\small
  The two-hump and three-hump energy densities are plotted in (a) with respect to $\omega$
  with other parameters being $m=1$, $s=1$, ${p}=0.25$, $v=-0.05$, $k=7$.
  The critical value for the two-hump energy densities transiting to the three-hump ones is
 $\omega= \frac{1035-\sqrt{856185}}{1920}\simeq 0.05713$.
  It is noted that the charge densities with the same parameters have only two humps, see (b).
}
\label{T00_three_hump}
\end{figure}

All we have discussed above is about the standing wave
(\ie the velocity $V=0$)
from which we can obtain the moving wave (\ie $0< V <1$)
by the Lorentz boost (see Eqs.~\eqref{eq:movingwave} and \eqref{eq:boost2})
in terms of variable of $\phi$ as follows
\begin{equation} \label{eq:boost1}
\tanh \phi = V, \;\; \cosh\phi=\gamma,
\;\; \sinh\phi = \gamma V, \;\;
\cosh\frac{\phi}2 = \sqrt{\frac{\gamma+1}{2}}, \;\;
\sinh\frac{\phi}2 = \sqrt{\frac{\gamma-1}{2}},
\end{equation}
where $\gamma := \tfrac{1}{\sqrt{1-V^2}}$ is the Lorentz factor.
The resulting relation between the right moving wave denoted by $\vPsi^\text{mw}$ and the standing wave denoted by $\vPsi^\text{sw}$ in
Eqs.~\eqref{localsol1} and \eqref{phaseexp} reads \cite{ShaoTang2005,XuShaoTang2013}
\begin{equation}\label{eq:movingwave}
\vPsi^{\text{mw}}(x, t) =
\vS
 \vPsi^{\text{sw}}(\tilde{x},\tilde{t}),
 \quad
  \vS := \bmb{    \begin{matrix}
      \cosh\frac{\phi}2 & \sinh\frac{\phi}2 \\ \sinh\frac{\phi}2 & \cosh\frac{\phi}2
    \end{matrix} },
\end{equation}
where
\begin{equation}
  \label{eq:boost2}
  \bmb{
  \begin{matrix}
  t \\
  x
  \end{matrix}
  }
  =
  \vLambda
  \bmb{
  \begin{matrix}
  \tilde{t} \\
  \tilde{x}
  \end{matrix}
  },
  \quad
  \vLambda := \bmb{
    \begin{matrix}
      \cosh \phi & \sinh \phi \\ \sinh \phi & \cosh \phi
    \end{matrix}},
\end{equation}
is the so-called Lorentz transformation
between the moving frame $(x,t)$ and the rest frame $(\tilde{x},\tilde{t})$.
Combining Eqs.~\eqref{jvector}, \eqref{localsol1}, \eqref{imag} and \eqref{eq:movingwave}
yields
\begin{equation}
j^0[\vPsi^{\text{mw}}](x, t) = \gamma j^0[\vPsi^{\text{sw}}](\tilde{x},\tilde{t}). \label{j0-mw}
\end{equation}
Moreover, it is straightforward to show that
\[
\partial^\mu = \Lambda^\mu_{\tilde{\mu}}\partial^{\tilde{\mu}},  \;\;
\vS^\dag \vgamma^0 = \vgamma^0 \vS^{-1}, \;\;
\vS^{-1} \vgamma^\mu \vS = \Lambda^\mu_{\tilde{\mu}}\vgamma^{\tilde{\mu}}, \;\;
\eta^{\mu\nu} = \Lambda^\mu_{\tilde{\mu}}\Lambda^\nu_{\tilde{\nu}} \eta^{\tilde{\mu}\tilde{\nu}},  \;\;
\]
where $\Lambda^\mu_{\tilde{\mu}}$ is the $(\mu,\tilde{\mu})$ entry of $\vLambda$ in Eq.~\eqref{eq:boost2},
and then
\begin{equation}
\label{eq:relation-tensor}
  T^{\mu\nu}[\vPsi^{\text{mw}}](x,t) = \Lambda^\mu_{\tilde{\mu}} \Lambda^\nu_{\tilde{\nu}}
  T^{\tilde{\mu}\tilde{\nu}}[\vPsi^{\text{sw}}](\tilde{x},\tilde{t}),
\end{equation}
from which we can readily verify
\begin{align}
T^{00}[\vPsi^{\text{mw}}](x, t) &= \gamma^2 T^{00}[\vPsi^{\text{sw}}](\tilde{x},\tilde{t}), \label{T00-mw}\\
T^{01}[\vPsi^{\text{mw}}](x,t) &= V \gamma^2 T^{00}[\vPsi^{\text{sw}}](\tilde{x},\tilde{t}), \label{T01-mw}
\end{align}
for Eqs.~\eqref{t00=t11=0}, \eqref{T01=0} and \eqref{eq:boost1}.
It is easy to see that,
Eq.~\eqref{j0-mw} (resp. \eqref{T00-mw}) implies
the charge (resp. energy) density for $\vPsi^{\text{mw}}$
has the same multi-hump structure as that for $\vPsi^{\text{sw}}$,
while the momentum density for $\vPsi^{\text{mw}}$
has the same multi-hump structure as the energy density for Eq.~\eqref{T01-mw}.

\subsection{$k=1$ and $\omega=m>0$}
\label{sec:hump:w=m>0}

When $k=1$, $\omega = m>0$ and
$(s,{p},{v})\in \mathcal{E}_1^-$,
the profile of the charge density $\rho_Q(x)$ has either one hump or two humps.
The reason is shown as follows. Recall from the discussion in Section \ref{sec:hump:m>w>=0} that
the number of humps in the charge density $\rho_Q(x)$
can be determined by the number of zeros of $\dd{\rho_Q(x)}{x}$ which has the form
\begin{equation}\label{dj0(x)_m=w}
  \dd{\rho_Q(x)}{x} =
  \frac{16 m^3 x }{(G(x))^2} \left(s+{v}- \frac{4 (s-{p})}{(4m^2x^2+1)^2}\right),
\end{equation}
for Eq.~\eqref{eq:R:w=m}.
From Eq.~\eqref{dj0(x)_m=w},
it is easy to see that
the charge density has three extreme points at $x=0,x=\pm
\frac{\sqrt{\sqrt{\frac{4(s-{p})}{s+{v}}}-1}}{2m}$ (\ie two humps) if $3s-4p - v < 0$,
otherwise has only one hump at $x=0$.
As for the energy density, combining Eqs.~\eqref{eq:cos-m=w}~\eqref{eq:R:w=m} and \eqref{T00-x}
leads to
\begin{equation*}
  \rho_E(x) = \frac{2m^2 (1-(2mx)^2)}{(s-p)((2mx)^2-1)^2+(v+p)((2mx)^2+1)^2},
\end{equation*}
then
\begin{equation*}
  \dd{\rho_E(x)}{x} = \frac{16 m^4 x \bmb{(s+{v})(4m^2x^2-1)^2- 4 (v+p)}}{\bmb{(s-p)((2mx)^2-1)^2+(v+p)((2mx)^2+1)^2}^2},
\end{equation*}
from which we have that
the energy density  $\rho_E(x)$
has two humps
at $x = \pm\frac{\sqrt{1-2\sqrt{\tfrac{v+p}{s+v}}}}{2m}$ if $3v+4p-s>0$,
otherwise has only one hump at $x=0$.
From Eqs.~\eqref{j0-mw} and \eqref{T00-mw},
we have the charge or energy density for the moving wave
has the same multi-hump structure as that for the standing wave as shown above.
According to Eq.~\eqref{T01-mw},
the momentum density for the moving wave
also has the same multi-hump structure as the energy density.

\section{Conclusion}
\label{sec:conclusion}

In this study, the NLD solitary waves under the linear combined self-interaction
to the power of the integer $k+1$ have been analytically derived and
the multi-hump structure in the charge, energy and momentum densities
has been rigorously analyzed.
We have proved that
for a given integer $k$,
the number of solitary humps for the charge density is bounded above with $4$,
while that for the energy density is bounded above with $3$.
Besides the two-hump structure first reported in \cite{ShaoTang2005},
the three-hump and four-hump charge densities have been observed.
We have also proved that
the four-hump charge density can only exist in the situation of higher nonlinearity,
\ie $k\in\{5,6,7,\cdots\}$,
while the three-hump one can appear in the situation of $k\in\{3,4,5,\cdots\}$.
The three-hump energy density which can only occur in the
situation of $k\in\{3,5,7,\cdots\}$ has also been pointed out.
It has been shown that the momentum density has
the same multi-hump structure as the energy density.
Our analysis has further revealed that,
the linear combined self-interaction in which
$p\neq 0$
as well as at least one of $s,v$ is not zero
is crucial for generating more than two humps (resp. one hump) in the charge (resp. energy) density.
Actually,
under the pure scalar self-interaction (\ie $s\neq 0$, $p=v=0$),
the charge density can be either one-hump or two-hump
while the energy density can only be one-hump;
under the pure vector self-interaction (\ie $v\neq 0$, $s=p=0$),
both the charge density and the energy density have only one hump;
under the linear combination of the scalar and vector self-interactions (\ie $s\neq 0,v\neq 0$, $p=0$),
the charge density can be either one-hump or two-hump
while the energy density can only be one-hump;
no physical solutions exist under the pure pseudoscalar self-interaction (\ie $p\neq 0$, $s=v=0$).
In addition, when $k=1$ and $\omega=m>0$, the NLD solitary wave with polynomial decay exists
and to our knowledge,
it has not been reported before this work.

\section*{Acknowledgments}
The research of this work was partly supported by grants from the
National Natural Science Foundation of China (Project Nos. 10925101,
11101011 and 91330110) and the Specialized Research Fund for the
Doctoral Program of Higher Education (Project No. 20110001120112).


\end{document}